\documentclass[11pt]{article}

\usepackage[margin=1in]{geometry}

\usepackage{amsmath}
\usepackage{amsthm}
\usepackage{amsfonts}
\usepackage{booktabs}
\usepackage{bbm}
\usepackage{natbib}
\usepackage{graphicx}
\usepackage{subfigure}
\usepackage[font=small]{caption}
\usepackage{xcolor}
\usepackage{enumitem}
\usepackage{amssymb}
\usepackage{comment}
\usepackage{url}
\usepackage{lineno}

\theoremstyle{definition}
\newtheorem{definition}{Definition}

\theoremstyle{plain}
\newtheorem{lemma}{Lemma}

\theoremstyle{plain}
\newtheorem{theorem}{Theorem}

\usepackage{tikz}
\usetikzlibrary{matrix,decorations.pathreplacing,calc}

\usepackage[plain,noend]{algorithm2e}

\def\T{{ \mathrm{\scriptscriptstyle T} }}

\renewcommand{\P}{\text{pr}}
\def\indic{\text{I}}	

\def\glt{generalized lower triangular{} }
\def\N{\mathbb{N}}
\def\mhprob{\pi}

\newcommand{\Hm}{H_m} 
\newcommand{\gammaPB}{\gamma}
\newcommand{\alphaPB}{\tilde{\alpha}}

\newcommand{\pq}[1]{q_{#1}}
\newcommand{\Lq}[1]{L_{#1}}
\newcommand{\jr}{r}
\newcommand{\qtilde}{\tilde{q}}

\newcommand{\SFS}[1]{\textcolor{green}{#1}} 
\newcommand{\commentS}[1]{\footnote{\textcolor{blue}{#1}}}  
\newcommand{\commentSold}[1]{}  

\graphicspath{{art/}}

\title{The generalized lower triangular process prior}

\author{%
Antonio Canale\thanks{Department of Statistical Sciences, University of Padova,
Via Cesare Battisti 241, 35121 Padova, Italy; and School of Mathematics,
University of Edinburgh, King's Buildings, EH9 3DF, Edinburgh, UK.
Email: \texttt{antonio.canale@unipd.it}.}
\and
Sylvia Fr\"uhwirth-Schnatter\thanks{Institute for Statistics and Mathematics,
WU Wien, Vienna, Austria. Email: \texttt{sfruehwi@wu.ac.at}.}
}
\date{}

\begin{document}

\maketitle

\begin{abstract}
In the analysis of multivariate data, factor models represent a powerful technique for both reducing dimensionality and facilitating  qualitative understanding of the latent determinants governing the observed data. Interpretability, however, is often hindered by the non-identifiability of factor loading matrices due to rotational invariance. A prevailing strategy to address this issue is to consider positive lower triangular structures. A relaxation of the latter condition has been recently proposed in the literature with the so called generalized lower triangular structure. However, its Bayesian implementation presents substantial challenges, requiring complex reversible-jump algorithms and making prior elicitation difficult because of the limited interpretability of the model parameters. In this paper, we introduce a generalized lower triangular process prior that allows both global and within-component sparsity structures for learning both the number of factors and possible sparsity structures within the vector of observations. The proposed approach  provides straightforward prior parameters elicitation exploiting possibly different prior information on the rank and sparsity characteristics.  We also explore connections with the Indian buffet process, providing further insight into the sparsity structure induced by the proposed prior. Posterior computation can be performed resorting to a Gibbs sampler with Metropolis-Hastings moves. The efficacy of the proposed method is demonstrated through comprehensive simulations and the analysis of  the Big Five Personality Test data.
\end{abstract}


\section{Introduction}
\label{sec:intro}
Factor models provide a convenient representation of the variability of $m$ dimensional random vectors by decomposing the $m\times m$ variance matrix $\Omega$ as  $\Omega = \Lambda \Lambda^\top+\Sigma,$
 where $\Lambda$ is a $m \times H$ matrix with typically $H \ll m$ and $\Sigma$ is a $m \times m$ diagonal matrix with positive entries $\sigma^2_{i}$ with $i  = 1, \dots, m$. 
This decomposition is equivalent to model each observation ${y}_t = (y_{1t}, \dots, y_{mt})^\top$ ($t=1, \dots, T$) as a linear combinations of independent latent Gaussian factors 
$\eta_{t}=(\eta_{1t},
\ldots, \eta_{Ht})^\top$ of dimension $H$
and an additive error term, i.e. 
\begin{align}
\label{eq:fact-mod}
y_t=\Lambda \eta_t +\epsilon_t, \qquad  \epsilon_t \sim N_m(0, \Sigma), \qquad {\eta_{t}} 
\sim 
N_{{H}} 
(0, I).
\end{align}
The factor analytic decomposition is thus condensing the available information into independent aggregates of lower dimension which could be potentially meaningful. 
Given this perspective, factor models find widespread use in various applications to handle datasets with numerous observed variables, believed to capture a smaller set of latent variables. This is particularly successful when the latter can be easily interpreted, such as in the case of individual skills and behaviors in psychology  \citep{fabrigar1999}, economics \citep{heckman2006}, genomics \citep{stein2018enter}, or ecology \citep{ovaskainen2017make}.

One of the crucial aspects associated to interpretability centers around determining the actual number of latent factors  $H$. The Bayesian approach offers an elegant mechanism to learn this quantity directly from the data, with notable contributions from \citet{Bhattacharya2011}, \citet{legramanti2020}, and \citet{fruehwirth2023GCUSP}.
Moreover, advancing the qualitative interpretation of Bayesian factor models beyond the inherent parsimony of the low-dimensional representation is achievable through the incorporation of shrinkage priors for the elements $\lambda_{ij}$ of $\Lambda$. Within a sparse loading matrix $\Lambda$, the impact of a specific factor is confined to a subset of the $m$ observed variables, facilitating interpretability by, for instance, grouping observed variables into dependent blocks or identifying variables uncorrelated with others. Recent years have witnessed numerous contributions in the realm of sparse Bayesian factor analysis. Analogous to variable selection shrinkage, these approaches broadly categorize into methods employing continuous shrinkage priors or point mass mixture priors, also called  spike-and-slab priors. Noteworthy instances of continuous shrinkage priors in Bayesian factor models encompass works by \citet{zhao2016}, \citet{rockova2017}, \citet{kastner2019}, and \citet{legramanti2020}. Conversely, spike-and-slab priors, adopted since the pioneering work of \citet{west2003}, have been explored by \citet{carvalho2008genes}, \citet{conti2014},  \citet{kaufmann2019} and \citet{fruehwirth23BA}, and \citet{schiavon2022}, and \citet{canaleschiavonstolf2026}.

Interpretation of the latent factors through their loading matrix, however, is often hindered by the rotational invariance property affecting $\Lambda$. It is well known, indeed, that  for any orthogonal matrix $P$ of dimension  $H \times H$, we have $\tilde{\Lambda} = \Lambda P$ with $\tilde{\Lambda}\tilde{\Lambda}^\top = \Lambda \Lambda^\top.$
This rotational ambiguity has received significant attention in the fields of econometrics, statistics, and machine learning. The most widely accepted approach for addressing rotational invariance is to consider positive lower triangular loading matrices. Examples of studies adopting this condition include \citet{geweke1996}, \citet{lopes2004}, and \citet{west2003}. 
Positive lower triangular structures have been generalized by \citet{fruehwirth2018} and \citet{fruehwirth2023GLT} via the so called generalized lower triangular structure with  \citet{fruehwirth23BA} introducing a spike-and-slab prior for Bayesian inference under these settings. However, its implementation present challenges, necessitating complex reversible jump algorithm. Additionally, prior specification may be not straighforward due to difficulties in terms of disentangling the local sparsity and the latent rank learning. 

To solve these issues, in this paper, we introduce a generalized lower triangular process prior that, consistently with the recent stream of literature on sparse Bayesian factor models, is able to jointly learn the latent factor dimension and the factor loading's local sparsity structure. Additionally, it ensures identifiability of the factor loading matrix, thereby permitting qualitative interpretability of results without the need for extensive post-processing \citep{papastamoulis2022,poworoznek}.
Our method is grounded in a probabilistic generative formulation of the generalized lower triangular structure, evocative of the Indian buffet process \citep{griffiths2011}. This formulation effectively disentangles the global shrinkage targeting the learning on the number of factors from the local shrinkage governing each factor loading.  Our proposed approach allows for straightforward elicitation of the prior hyperparameters and ease in posterior computation.

\section{Generalized lower triangular process prior}
\label{sec:GLTstructure}

\subsection{Notation and definitions}

Let  $\delta_{ij} = \indic\{\lambda_{ij} \neq 0 \}$ be a binary variable indicating if the 
element in row $i$ and column $j$ of $\Lambda$ is different from zero or not. Let $\Delta$ be the indicator matrix of the same dimension of $\Lambda$ containing the $\delta_{ij}$s.  We introduce the following definition:

\begin{definition}
For each column of $\Lambda$, the pivot row $\ell_j$ is defined as the row index of the first non-zero factor loading in column $j$. The element $\lambda_{\ell_jj}$ is called the leading factor loading of column~$j$.    
\end{definition}

The lower triangular structure assumes that $\ell_j=j$. In the \glt{} structure \citep{fruehwirth2023GLT} the pivots are still ordered, i.e.  $\ell_1 <  \ell_2 < \dots < \ell_H$ but are not necessarily constrained to be equal to their column index.  Notably, the indicator matrix $\Delta$ with \glt structure is a special case of a left-ordered binary matrix commonly adopted in features allocation models \citep{griffiths2011} as further discussed in Section \ref{sec:GLTIBP}. 

Under a \glt{} structure, let  $W_j = \ell_{j} - \ell_{j-1}-1$ be the number of rows between two {adjacent}
pivots with the convention that $\ell_{0}=0$ and $W_1 = \ell_1 -1$. 
While the pivots are quantities naturally associated with the columns of the  matrices $\Lambda$ and $\Delta$, one can equivalently characterize the rows. Define $K_i$ recursively as the maximum between $K_{i-1}$ and the largest column index corresponding to a nonzero loading in row $i$, with the convention that $K_0=0$. It follows that knowledge of any one of the three sequences---$\ell_j$, $W_j$, or $K_i$---fully determines the other two. Figure~\ref{fig:sketch} reports a pictorial representation of these different quantities.

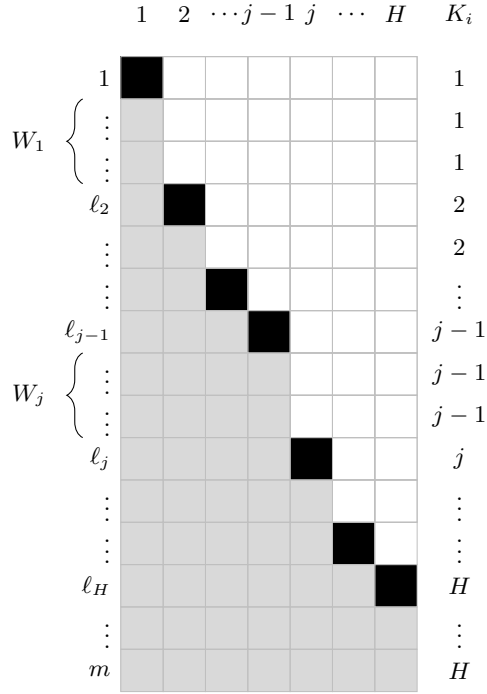
\begin{figure}[h!]
\centering
\begin{tikzpicture}[scale=0.8, font=\footnotesize]

\def\s{0.7}   
\def\ncols{7}  
\def\nrows{11} 

\foreach \r in {1,...,11}{
  \foreach \c in {1,...,7}{
    \draw[gray!50, line width=0.3pt] (\c*\s,-\r*\s) rectangle ++(\s,-\s);
  }
}

\fill[gray!30] (1*\s,-2*\s) rectangle ++(\s,-14*\s);
\fill[gray!30] (2*\s,-4*\s) rectangle ++(\s,-12*\s);
\fill[gray!30] (3*\s,-6*\s) rectangle ++(\s,-10*\s);
\fill[gray!30] (4*\s,-7*\s) rectangle ++(\s,-9*\s);
\fill[gray!30] (5*\s,-10*\s) rectangle ++(\s,-6*\s);
\fill[gray!30] (6*\s,-12*\s) rectangle ++(\s,-4*\s);
\fill[gray!30] (7*\s,-13*\s) rectangle ++(\s,-3*\s);

\fill[black] (1*\s,-1*\s) rectangle ++(\s,-\s);
\fill[black] (2*\s,-4*\s) rectangle ++(\s,-\s);
\fill[black] (3*\s,-6*\s) rectangle ++(\s,-\s);
\fill[black] (4*\s,-7*\s) rectangle ++(\s,-\s);
\fill[black] (5*\s,-10*\s) rectangle ++(\s,-\s);
\fill[black] (6*\s,-12*\s) rectangle ++(\s,-\s);
\fill[black] (7*\s,-13*\s) rectangle ++(\s,-\s);

\foreach \r in {1,...,15}{
  \foreach \c in {1,...,7}{
    \draw[gray!50, line width=0.3pt] (\c*\s,-\r*\s) rectangle ++(\s,-\s);
  }
}

\node at (1.5*\s,0) {$1$};
\node at (2.5*\s,0) {$2$};
\node at (3.5*\s,0) {$\cdots$};
\node at (4.5*\s,0) {$j-1$};
\node at (5.5*\s,0) {$j$};
\node at (6.5*\s,0) {$\cdots$};
\node at (7.5*\s,0) {$H$};

\node[left] at (1*\s,-1.5*\s) {$1$};
\node[left] at (1*\s,-2.5*\s) {$\vdots$};
\node[left] at (1*\s,-3.5*\s) {$\vdots$};
\node[left] at (1*\s,-4.5*\s) {$\ell_2$};
\node[left] at (1*\s,-5.5*\s) {$\vdots$};
\node[left] at (1*\s,-6.5*\s) {$\vdots$};
\node[left] at (1*\s,-7.5*\s) {$\ell_{j-1}$};
\node[left] at (1*\s,-8.5*\s) {$\vdots$};
\node[left] at (1*\s,-9.5*\s) {$\vdots$};
\node[left] at (1*\s,-10.5*\s) {$\ell_j$};
\node[left] at (1*\s,-11.5*\s) {$\vdots$};
\node[left] at (1*\s,-12.5*\s) {$\vdots$};
\node[left] at (1*\s,-13.5*\s) {$\ell_H$};
\node[left] at (1*\s,-14.5*\s) {$\vdots$};
\node[left] at (1*\s,-15.5*\s) {$m$};

\node at (9*\s,-1.5*\s) {$1$};
\node at (9*\s,-2.5*\s) {$1$};
\node at (9*\s,-3.5*\s) {$1$};
\node at (9*\s,-4.5*\s) {$2$};
\node at (9*\s,-5.5*\s) {$2$};
\node at (9*\s,-6.5*\s) {$\vdots$};
\node at (9*\s,-7.5*\s) {$j-1$};
\node at (9*\s,-8.5*\s) {$j-1$};
\node at (9*\s,-9.5*\s) {$j-1$};
\node at (9*\s,-10.5*\s) {$j$};
\node at (9*\s,-11.5*\s) {$\vdots$};
\node at (9*\s,-12.5*\s) {$\vdots$};
\node at (9*\s,-13.5*\s) {$H$};
\node at (9*\s,-14.5*\s) {$\vdots$};
\node at (9*\s,-15.5*\s) {$H$};

\node at (9*\s,0) {$K_i$};

\draw[decorate,decoration={brace,mirror,amplitude=6pt}]
  (0.1*\s,-8*\s) -- (0.1*\s,-10*\s)
  node[midway,left=10pt,align=center] {$W_j$};

\draw[decorate,decoration={brace,mirror,amplitude=6pt}]
  (0.1*\s,-2*\s) -- (0.1*\s,-4*\s)
  node[midway,left=10pt,align=center] {$W_1$};
\end{tikzpicture}
\caption{
Schematic representation of the constrained support of $\Delta$.
For each column $j$, the pivot is located at row $\ell_j$ (black), entries below the pivot are admissible (gray), 
$K_i$ denotes the number of pivots above or at row $i$, and $W_j$ denotes the number of rows between two consecutive pivot locations.}
\label{fig:sketch}
\end{figure}


In the next section we define a stochastic mechanism to draw  the  \glt binary matrix $\Delta$ and later define the prior for the $\lambda_{ij}$ conditionally on  $\delta_{ij} =1$.

\subsection{Generalized lower triangular  process}
\label{sec:GLTprocess}

  We first introduce a stochastic   process for a matrix $\Delta_\infty$ with infinite rows and columns and then truncate the number of rows to $m$ to obtain the prior distribution for {$\Delta$}.
  Consistently, we assume that $\{\ell_j\}_{j\geq 1}$,  $\{W_j\}_{j\geq 1}$, and  $\{K_i\}_{i\geq 1}$, are infinite 
  {random} sequences.

For $i  \in \N$ where $\N$ is the set of natural numbers, we let $K_i$ be an increasing stochastic process  with independent unitary jumps at jump times $\ell_j$. Let the jump probability, i.e. $\P( K_i = K_{i-1}+1 \mid \nu) = \nu$. 
Notably,  the marginal distribution for the single $K_i$ given $\nu$ is Binomial with $i$ trials and success probability $\nu$, i.e.  ${K_i 
\mid \nu}
\sim \mbox{Bin}(i, \nu).$ Indeed, drawing a sequence of $K_i$s is equivalent to draw a sequence of independent Bernoulli random variables with success probability $\nu$ and then taking 
{their cumulative sum.}
	
Consistent with this, each $W_j$ represents the waiting time between jump $j-1$ and jump $j$ and is a geometric random variable with  probability equal to $\nu$. Thus,  the joint probability distribution of the first $h$ waiting times $W_j$s is equal to
\begin{eqnarray} \label{wait}
\P(W_1=w_1, \dots, W_{h}=w_h \mid \nu) =\nu^{h}  
{\prod_{j=1}^{h}} 
 (1-\nu)^{w_j} =  \nu^{h}  (1-\nu)^{\sum_{j=1}^{h} w_j}.
\end{eqnarray}

Given the pivots, the \glt structure forces the presence of  structural zeros above the generalized diagonal. Below the pivots, instead,  null entries may or may not occur. In accordance with the sparse factor models literature, we allow the possibility of further local sparsity below the pivots. However, differently from \citet{fruehwirth23BA}, we disconnect this local sparsity below the pivots from the mechanism tied to the pivots. This strategic disentanglement effectively separates the global shrinkage, which targets learning the number of factors, from the local shrinkage that governs each factor loading. As clarified later in the interpretation presented in Section \ref{sec:GLTIBP} and substantiated by empirical evidence in Section \ref{sec:simulation}, this approach facilitates prior calibration and improves posterior inference.

Specifically, we assume $ \P(\delta_{ij} = 1 \mid \ell_j, \tau_j) = \tau_{j} $ for $i>\ell_j$ where  $\tau_j$ are independent and identically distributed beta random variable, $\tau_j \sim \mbox{Beta}(\alpha \beta, \beta).$ 
Previous discussion allows us to introduce the following definition.

\begin{definition}[Generalized lower triangular process]
For $i  \in \N$,  let $K_i$ be an increasing stochastic process  with independent unitary jumps at jump times $\ell_j$ with jump probability $\P( K_i = K_{i-1}+1 \mid \nu) = \nu$.  For $i> \ell_j$ let $ \P(\delta_{ij} = 1 \mid \ell_j, \tau_j) = \tau_{j} $. For $j\geq 1$, let  $\tau_j \sim  \mbox{Beta}(\alpha \beta, \beta)$ and $\nu \sim \mbox{Beta}(a_\nu, b_\nu)$.
Then, the left-ordered infinite dimensional binary matrix $\Delta_\infty$ follows a \glt{} process. 
\label{def:gltprocess}
\end{definition}

Note that if we let the \glt{}  process running, it will eventually lead to a pivot $\ell_j=\sum_{l=1}^{j-1} {W}_{l}+j$ larger than $m$.
Knowing that in practice the largest pivot is bounded by $m$,  we run this process until $\ell_j > m$. Such a truncation allows to define a prior distribution on generalized lower triangular matrices with $m$ rows and a random number of columns. To this end, we introduce the following result.
\begin{lemma}
	Let $h \in \N$, such that  $\ell_1<\dots<\ell_h\leq m$ and $\ell_{h+1}>m$.
Then
\begin{equation}
	\P(\ell_1=l_1, \ldots, \ell_{h}=l_h,  \ell_{h+1}>m \mid  \nu ) =
	\nu^h (1-\nu)^{m-  h} ,
\end{equation}
and similarly
\begin{equation}
	\P(W_1=w_1, \ldots, W_{h}=w_h,  W_{h+1} >  m-\sum_{l=1}^{h} W_l - h \mid \nu ) =
	\nu^h (1-\nu)^{m-  h} .
\end{equation}
\label{lemma:W}
\end{lemma}

It  now clearly results that the number of columns $H$ of a \glt{} matrix $\Delta$ with $m$rows obtained by truncating $\Delta_\infty$ at the $m$th row, is a random variable implicitly obtained by running the process. We formalize this in the following definition.
\begin{definition}[Generalized lower triangular prior]
Let $\Delta_\infty$ be a left-ordered infinite dimensional binary matrix following a \glt process. Let $H \in \N$ the random number satisfying Lemma \ref{lemma:W}. Then the $m \times H$ random matrix $\Delta$ obtained by truncating $\Delta_\infty$ at row $m$ and column $H$, follows a finite \glt prior. 
\label{def:gltprior}
\end{definition}


To complete the prior specification for the factor model \eqref{eq:fact-mod}, we assume a diffuse distribution for the entries $\lambda_{ij}$ associated with $\delta_{ij}=1$, as well as for the idiosyncratic variances $\sigma_i^2$. To maintain analytical tractability, we adopt a simple structure henceforth, directing our attention towards evaluating the efficacy of the \glt{} prior. It is pertinent to note, however, that more complex distributions can be considered \citep[see][]{fruehwirth23BA}. Specifically, prior elicitation is completed by assuming 
\begin{align}
    (\lambda_{ij} \mid \delta_{ij}= 1 ) \sim N(0, \kappa \sigma_i^2),  \quad\quad \sigma^2_i\sim \mbox{I-Ga}(a_\sigma, b_\sigma),
    \label{eq:priorlambdasigma}
\end{align}
where $\kappa ,$ $a_\sigma$, and $b_\sigma$ are positive hyperparameters.

\subsection{Generalized lower triangular prior properties}

To facilitate prior elicitation, it is useful to clarify how the hyperparameters $\alpha, \beta, a_\nu,$ and $b_\nu$ control key structural features of the induced matrix $\Delta.$ 
For instance, the expected number of  columns is derived in the following Lemma.
\begin{lemma}
\label{lem:priornumberofcols}
    Let $\Delta$ be assigned a \glt prior. Then, 
\[
\mathbb{E}[H]  = \mathbb{E}[K_m] = m \frac{a_\nu}{a_\nu+b_\nu}.
\]
\end{lemma}
Lemma~\ref{lem:priornumberofcols} shows that the hyperparameters $a_\nu$ and  $b_\nu$ directly govern the prior expected number of columns of the loading  matrix, which grows linearly with the number of variables $m$. 
{Notably, under the prior of  \citet{fruehwirth2023GLT}, the corresponding expectation admits the large-$m$ approximation $ \mathbb{E} (H|\alphaPB, \gammaPB,
 \Hm) \approx  \alphaPB \gammaPB \log m - \alphaPB \gammaPB \psi(\gammaPB)$,
 where $\psi(x)=\Gamma'(x)/\Gamma(x)$ is the digamma function, see Appendix~\ref{UGLTprior} for details.}
In that formulation, the prior expected number  of columns is driven by the parameters $\alphaPB$ and $\gammaPB$, with $\alphaPB$ also impacting other structural properties of the matrix $\Delta$, as discussed below.

Another important feature to be studied is the prior expected sparsity level of $\Delta$. This is characterized in the following lemmas whose proofs are reported in Appendix~\ref{Proofs}.
\begin{lemma}
\label{lem:priorsparsity}
Let $\Delta$ be assigned a \glt prior. Let $N_m$ denote the total number of non-zero entries in  $\Delta$. Then, for every $h\geq1$
\[
\mathbb{E}\left[\frac{N_m}{mH}\mid \alpha,  H = h\right] = \frac{1}{m}+ \frac{\alpha}{\alpha+1}\frac{m-1}{2m}.
\]
\end{lemma}

Notably, under the \citet{fruehwirth2023GLT} prior a similar closed-form expression of Lemma \ref{lem:priorsparsity} exists and is a function of the parameter $\alphaPB$ (see Appendix~\ref{UGLTprior}). Hence,  under 
the \citet{fruehwirth2023GLT} prior, such parameter simultaneously regulates both the within-column  sparsity and the prior expected number of columns.
In our construction, by contrast, these two quantities are governed by distinct  sets of hyperparameters: $\alpha$ determines the within-column sparsity as discussed in Lemma \ref{lem:priorsparsity}, while  $(a_\nu,b_\nu)$ control the expected number of columns as discussed in Lemma \ref{lem:priornumberofcols}. 
This separation allows for clearer and more flexible prior elicitation, allowing practitioners 
to tune expected prior sparsity and latent dimension independently and does not require deep hierarchical constructions as required in  \citet{fruehwirth2023GLT}. 
As discussed in Section~\ref{sec:simulation},  this structural decoupling also yields concrete empirical advantages. One may ask the role of the $\beta$ hyperparameter, which is not appearing in previous lemmas. Clearly, this parameter is related to the prior variability of the \glt process as discussed in the next Lemma.

\begin{lemma}\label{lem:column_sparsity_var}
Let $\Delta$ be assigned a \glt prior. Let $A_j=m-\ell_j>0$ be the available entries below the $j$th pivot, 
$N_j=\sum_{i=\ell_j+1}^m \delta_{ij}$ the number of ones below  pivot $j$, and $
R_j=N_j/A_j$ a measure of conditional column sparsity. 
Then, 
\[
\mathrm{Var}(R_j\mid A_j)
= \frac{\alpha}{A_j(\alpha+1)^2}
+
\frac{\alpha(A_j-1)}{A_j(\alpha+1)^2\{(\alpha+1)\beta+1\}}.
\]
\end{lemma}

Lemma~\ref{lem:column_sparsity_var} shows that the conditional variance of the within-column sparsity ratio $R_j$
admits a natural decomposition into two components.
The first term reflects binomial sampling variability among the $A_j$ entries below the pivot and decreases as the number
of available entries increases. The second term, instead, captures the additional variability induced by heterogeneity across columns through the beta prior on $\tau_j$. For fixed $\alpha$, increasing $\beta$ concentrates the prior on $\tau_j$ around its mean $\alpha/(\alpha+1)$ and therefore reduces the prior variance. Notably, this variance does not depend on $\beta$ if $A_j=1$.

\subsection{A feature allocation interpretation}
\label{sec:GLTIBP}

The \glt process introduced in the previous section exhibits notable affinities with the Indian buffet process \citep{griffiths2011}. The Indian buffet process constitutes a stochastic process that defines a probability distribution over equivalence classes of sparse binary matrices characterized by a finite number of rows and an unbounded number of columns. This concept is often introduced by a culinary metaphor, where binary matrices symbolize a scenario in which $m$ customers (rows)  choose a random number of dishes (columns). The procedural framework proceeds as follows: the $m$ customers sequentially enter a restaurant, with the initial customer choosing a random number of dishes. Subsequently, the $i$-th customer selects dishes in proportion to the choices made by previous customers, augmented by a random number of new dishes. In the context of the Indian buffet process, the random number of new dishes is sampled from a Poisson distribution, while the selection of previous dishes follows suitably specified Bernoulli random variables.

Our \glt process can be interpreted with a similar metaphor. The main difference lies in the fact that each customer chooses at most one new dish with probability $\nu$. Looking at the iterative conditional distributions of a \glt process row by row, one can see that the conditional distribution of this $\nu$ given  $K_i$  is
\[
\nu \mid K_i \, \sim\,   \mbox{Beta}\left(a_\nu + K_i, b_\nu +i -K_i\right).
\]
The expectation of this distribution is given by $(a_\nu + K_i)(a_\nu + b_\nu + i)^{-1}.$ This expression implies that customer $i+1$ is expected to have a probability of choosing a new dish that is less than or equal to that of the previous customer. In the context of our motivating factor model, this signifies that the probability of adding a factor exhibits a stochastic decreasing behaviour with increasing $i$ consistently with current approaches in Bayesian infinite factor loadings models \citep{Bhattacharya2011,legramanti2020}.

Similarly, customer $i+1$ chooses the $j$-th dish among the $K_{i}$ already chosen with probability $\tau_j$ which conditional distribution is 
\[
\tau_j \mid \ell_j, d_{ij} \, \sim\,   \mbox{Beta}\left(\alpha \beta + d_{ij}, \beta + i-\ell_j -d_{ij}\right),
\]
where $d_{ij}$ is the number of customers choosing dish $j$ after its first choice,  i.e. $d_{ij} =  \sum_{l=\ell_j+1}^i \delta_{lj}$. Its expectation is  $
{(\alpha \beta + d_{ij})}  ( \alpha\beta + \beta + i-\ell_j)^{-1}.$

\section{Posterior computation}
\label{sec:posterior}

\subsection{Markov Chain Monte Carlo sampling scheme}

In this section, we present a Markov chain Monte Carlo scheme for sampling from  the posterior distribution of the parameters of model~\eqref{eq:fact-mod} under the prior specification introduced in Section~\ref{sec:GLTprocess}. To update the  indicator matrix $\Delta$, the sampler exploits the probabilistic construction of the \glt{} process to define tailored Metropolis--Hastings moves. Posterior  sampling proceeds by iterating the steps summarized below, with details of the  main updates provided in the following subsections.

The proposed methodology is implemented in the \texttt{R} package \texttt{glt},  with the most computationally intensive routines written in \texttt{C++} and  interfaced with \texttt{R}. The package implements the full posterior sampling  scheme described below and is available at  \texttt{github.com/tonycanale/glt}.

\begin{enumerate}
	\item Update $H$. See Section \ref{sec:mhK}.
	\item Update  the probability $\nu$ of starting a new column  from 
	\[
	\nu \mid H \, \sim\,   \mbox{Beta}\left(a_\nu + H, b_\nu +m -H\right).
	\]
	\item Update the pivots locations  given the number of columns $H$.  See  Section \ref{sec:mhGLT}.
	\item Update the column specific probabilities $\tau_j$ for $j = 1, \dots, H$ from
	\[
	{\tau}_j \mid \ell_j, \Delta \, \sim\,   \mbox{Beta}\left(\alpha \beta + d_j, \beta + m-\ell_j -d_j\right),
	\]
	where $d_{j} =  \sum_{i=\ell_j+1}^m \delta_{ij}$. 
	\item Update the $\delta_{ij}$ below the pivots, given the pivots and marginalizing with respect to $\Lambda$ and $\Sigma$ (column-wise).
\item Update the variances $\sigma_i^2$ from 
	\[
	\sigma_i^{-2} \mid - \sim \mbox{Ga} \left\{ a_\sigma  + \frac{T}{2}, b_\sigma + \frac{Q_i}{2}  \right\},
	\]
 where $Q_i =
{y_{i\cdot} ^\top y_{i\cdot}
- y_{i\cdot} ^\top
\eta_i^\delta 
[\kappa^{-1} I_{q_i} + (\eta_i^\delta) ^\top \eta_i^\delta ]^{-1}
(\eta_i^\delta) ^\top
y_{i \cdot}}$
	\item Update  the 
  factor loadings $\lambda_{ij}$ 
 (row-wise), 
 given $\Delta$, the latent factors and $\sigma_i^{2}$. 
 Clearly $\lambda_{ij}| \delta_{ij}=0$ is equal to zero and only the $q_i$ non-zero elements
 $\lambda_{ij}| \delta_{ij}=1$ in row $i$, 
 collected in the $q_i$-dimensional vector $\lambda_{i\cdot}^\delta$,
  are samples from a Gaussian posterior of dimension $q_i$:
	\[
 \lambda_{i\cdot}^\delta
 \mid   - \, \sim \, N_{q_i} 
 \left(   
 B_{i}^\delta m_i^{\delta},
 B_{i}^\delta \sigma_i^2
 \right),
	\]
    where $B_{i}^\delta$ and $m_i^{\delta}$ are defined as: $B_{i}^\delta =(\kappa I_{q_i} + {\Xi_i^{\delta}}^\top \Xi_i^{\delta})^{-1}$ with $I_{q_i}$ the identity matrix of dimension $q_i$ and  $\Xi_i^\delta$ the matrix defined as $\Xi_i^\delta = (\eta_1, \dots, \eta_T)^\top$ if no element in row $i$ of $\Delta$ is restricted to zero or, similarly, taking $(\eta_1, \dots, \eta_T)^\top$ and deleting all columns $j$ with $\delta_{ij}=0$; $m_i^{\delta} = {\Xi^\delta_i}^\top \tilde{y}_i$ with $\tilde{y}_i = (y_{i1}, \dots, y_{iT})^\top.$
	\item Update the latent factors $\eta_{t}$ for $t=1, \ldots,T$ conditionally on $\Lambda$ from
	\begin{equation*}
	\eta_{t}\mid-  \sim N_{H}\big\{(I_H+
{\Lambda}^\T \Sigma^{-1}{\Lambda})^{-1} {\Lambda}^\T \Sigma^{-1} y_{t},\, (I_H+{\Lambda}^\T \Sigma^{-1} {\Lambda})^{-1}
 \big\}.
	\end{equation*} 
\end{enumerate}

\subsection{Update the number of active columns}
\label{sec:mhK}

The number of active columns $H$ is updated with a random walk Metropolis-Hasting step. We separately detail the case in which the proposal increases or decreases the current value of $H$ separately. {Whenever possible, we will choose between the two updates with probability $0.5$. If  
$\ell_H = m$, then there is no space to add an additional pivot and we propose to decrease $H$  with probability $1$. Similarly, when $H=0$, then we propose to increase $H$  with probability $1$.}

\commentSold{Comments Sylvia:
\begin{itemize}
    \item According to Lemma~3, it is possible that
$\ell_{H} = m-1$ and $\ell_{H+1} = m$ and we have to allow an increase of $H$ in this case, even if it leads to a spurious column.
\item Definition~3 does not rule out the case $H=0$,
even it is very unlikely for realistic values of $m$. Nevertheless we need to allow it, but in this case the only proposal is to increase $H$.
\item Hence, for each proposal, there is a  case, where the proposal ratio is not equal to 1, see below.  
\end{itemize}}

First, consider to propose $H^* = H + 1$, that is add a new column 
${\delta^*_{\cdot H^*}}$
with pivot $\ell_{H^*}^*$ drawn uniformly between  $\ell_{H}+1$ and $m$. Notably 
$\P(H)/\P(H^*) 
= (1-\nu)/\nu$.
Below pivot $\ell_{H^*}^*$, fill the new column sampling the indicators row-wise with probability
\begin{equation}
\P({\delta^*_{iH^*}}
= 1 ) = \frac{B(\alpha \beta +S_{i-1} +1, \beta + F_{i-1} )}{B(\alpha\beta + S_i, \beta + F_i)},
\label{eq:ibplikedelta}
\end{equation}
where $S_i = \sum_{l=\ell_{H^*}^*+1}^{i} 
{\delta_{lH^*}^*}$  and $F_i = \sum_{l=\ell_{H^*}^*+1}^{i} (1-{\delta_{lH^*}^*})$.

Equivalently, we can sample the $\delta_{iH^*}$ below the pivot in block by sampling $d_{H^*}$ from the multinomial distribution with probabilities:
\begin{equation}
\text{pr}(d_{H^*} = d ) = \binom{m-\ell_{H^*}^*}{d} \frac{B(\alpha \beta +d, \beta + {m-\ell_{H^*}^*} -d )}{B(\alpha\beta, \beta)}, 
\label{eq:ibplikedelta2}
\end{equation}
for 
 $d=0, \dots, {m-\ell_{H^*}^*}$ 
and then assign $d_{H^*}$ ones randomly in the entries below  pivot $\ell_{H^*}^*.$
Finally, 
 add  
$\eta_{H^*t}^*$ to the  new factor $ \eta_t^*$
by sampling $\eta_{H^*t}^* \sim N(0,1)$ for $t \in 1, \dots, 
T$.
Let $\Delta^*$ be the new indicator matrix which is identical to $\Delta$ for the first $H$ columns, but has an additional $H^*$-th column {$\delta_{\cdot H^*}^*$} filled with the newly sampled {$\delta^*_{iH^*}$.}

The new column is accepted with probability  $\mhprob = \min (1, {R^+})$ where
\[
R^+ =
\prod_{i: 
\delta^*_{iH^*}
=1}
\prod_{t =1}^T \frac{
	\P(y_{it} \mid 
 {\Delta^*, \eta_t^*)} 
}{
	\P(y_{it} \mid 
{\Delta, \eta_t)}  
} \times 
\frac{
	\P(H^*)\P(
 {\Delta ^*} 
 )\prod_{t =1}^T \P(\eta_t^*)
}{
	\P(H)\P(
 {\Delta} 
 )\prod_{t =1}^T \P(\eta_t)
} \times  \frac{
	1
}{
	\P(
{\delta_{\cdot H^*}^*} 
 ) 
 \prod_{t=1}^T\P(
{\eta_{H^* t}^*} 
 ) 
}  \times p_{+}, 
\]
where the first element of the product is the ratio of the (marginal) likelihoods, the second is the prior ratio, the third is the ratio of the proposals, and $p_{+}=0.5$, if $1\leq H\leq m-1$ or $p_{+}=1$, if $H=0$.

Now consider to propose $H^* = H - 1$, that is drop the last column {$\delta_{\cdot H^*}$} from the factor loading matrix $\Lambda$ and the related entries 
{$\eta_{H ^* t}$}
in the factors $\eta_t$. 
This move is accepted with probability 
$\mhprob = \min (1,{R^-})$ where
\[
R^- = 
\prod_{i: \delta_{iH}=1 
} 
\prod_{t =1}^T \frac{
	\P(y_{it} \mid 
 \Delta^*, \eta_t^*)
}{
	\P(y_{it} \mid 
 \Delta, \eta_t)
} \times  \frac{
	\P(H^*)\P(\Delta ^*)\prod_{t =1}^T \P(\eta_t^*)
}{
	\P(H)\P(
 \Delta
 )
 \prod_{t =1}^T \P(\eta_t)
} \times
\P(
{\delta_{\cdot H^*}} 
) \prod_{t=1}^T\P(
 {\eta_{H^* t}} 
) \times  p_{-}, 
\]
where $p_{-}=0.5$, if $1\leq H\leq m-1$ or $p_{-}=1$, if $H=m$.

Simplifying one gets that, $R^+$ and $R^{-}$ are simply equal to the likelihood ratios of the elements indexed by $i$ and $j=\max\{H, H^*\}$ such that  $\delta_{ij}=1$ multiplied by $\nu/(1-\nu) p_{+}$ or $(1-\nu)/\nu p_{-}$, respectively.

\subsection{Update the pivots' locations}
\label{sec:mhGLT}

Also the pivots' locations are updated through a simple Metropolis-Hastings move which exploits the following theorem: 
\begin{theorem}
\label{thm:uniformpivots}
Let $\Delta$ be a random binary matrix with $m$ rows sampled from a finite \glt prior. Then, the probability of any sequence $\ell_1, \dots, \ell_{H}$ with the same number of pivots $H$ is uniform over all 
$m! \{(m-H)! H!\}^{-1}$ sequences. 
\end{theorem}

As already discussed in Section~\ref{sec:GLTstructure}, drawing the sequence of $K_i$s is equivalent to draw a sequence of independent Bernoulli random variables $k_i$ with success probability $\nu$ and then taking the cumulative sum of them.  Conditionally on their total sum $H$,  we consider as proposal in a Metropolis-Hastings step, a random permutation of the sequence of the $k_i$s.  This proposal is uniform over the space of permutations  and  invariant under the prior thanks to Theorem 1. Hence the prior/proposal ratio for the new pivots $(\ell_1^*, \ldots, \ell_H^*)$  is simply equal to 1.

Conditionally on the obtained structure,  in each column $j$ whenever $\ell_j \neq \ell_j^*$, some $\delta_{ij}^*$ between $\ell_j$ and $\ell_j^*$ become structural zeroes, while other $\delta_{ij}^*$ that previously were structural zero are eligible to host a one.
Similarly, the indicator $\delta_{ \ell_j,j}$  corresponding to the old pivot is turned into a structural zero,  if $\ell_j < \ell_j^*$, and is eligible to host a one,  if $\ell_j > \ell_j^*$.

We introduce in what follows a way to propose the new indicators such that 
the prior/proposal ratio is equal to one. This is straightforward when  $\ell_j < \ell_j^*$, since the proposed indicators 
$\delta^*_{i,j}$ for $i=\ell_j, \ell_j +1, \dots, \ell^*_j$ are deterministically equal to zero and all other indicators are held fixed at the corresponding indicators of 
$\delta_{\cdot j}$. 

For $\ell_j > \ell_j^*$, instead, the new indicators $\delta^*_{i,j}$ for $ i =  \ell_j^*+1, \ldots, \ell_j$ are randomly proposed as follows.
Let ${d_{a,j}} =\sum_{i={{\ell_{j}^*+1}}}^{\ell_j} \delta^*_{i,j}$ be the random number of additional non-zero elements in the factor loading matrix, i.e. $d_j^*=d_j+d_{a,j}$. We propose  ${d_{a,j}}  
\in \{0,  \ldots, \ell_j-\ell_{j}^*\}$  with
\begin{eqnarray}\label{condda}
\P(d_{a,j} = d \mid  d_j, \ell_{j}^{*}, \ell_j)= 
  \displaystyle
   \binom{\ell_j - \ell_j^*}{d}
   \frac{  B(\alpha\beta + d_j + d, \beta + m-\ell_j^* - d_j - d)}{  
   B(\alpha\beta + d_j, \beta + m-\ell_j - d_j )}.
\end{eqnarray}
Then, new indicators $\delta^*_{i,j}$ are proposed by assigning ${d_{a,j}}$  ones randomly between  $\ell_{j}^*+1$ and $\ell_j$.

\section{Simulation Study}
\label{sec:simulation}

We assess the finite-sample performance of the proposed \glt{}  prior
and compare it with the unordered generalized lower triangular approach of
\citet{fruehwirth23BA}, implemented through its reversible-jump MCMC sampler.

From a computational perspective, the simulations are conducted using the \texttt{R} package \texttt{glt}. A direct comparison of computing times with the \citet{fruehwirth23BA} approach would, however, be difficult to interpret fairly. The latter is implemented in \texttt{MATLAB} through code developed by the authors, which is not publicly distributed, and our numerical experiments with that implementation were necessarily run on a different computing platform. Differences in programming language, low-level implementation, and hardware would therefore be confounded with differences in the underlying algorithms \citep{kriegel2017black}.

Consistently, we do not report formal wall-clock-time comparisons. As a rough indication of the computational cost of the proposed \glt,  $10{,}000$ iterations required about three minutes for the largest configurations and around 20 seconds for the smallest ones on a MacBook Pro with an Apple M5 processor and 48\,GB of RAM.

For comparability, both samplers were initialized at the true number of factors and the true pivot configuration and were run for $10{,}000$ iterations discarding the first $6{,}000$ as burn-in.

Data were generated from the Gaussian factor model in \eqref{eq:fact-mod}, under four designs chosen to represent distinct sparsity regimes. For each design, we considered $m \in \{20,50,100\}$ and sample sizes $T \in \{100, 200\}$. The true number of factors varied with the scenario and dimension. Each configuration was replicated 20 times. 
Scenario 1 considered a generalized lower triangular structure with irregularly spaced pivots and dense columns below the pivots. Specifically, the first pivot was fixed at 1, while the remaining pivots were randomly generated between 2 and $m/5$. Conditional on the pivot locations, the below-pivot entries were independently set to one with probability $0.95$, yielding configurations characterized by a relatively small number of dense factors. Both the pivot locations and the sparsity patterns were regenerated at each replicate. In this scenario $(m,H) \in \{(20,3), (50,8), (100,10)\}$.
Scenario 2 retained the same generalized lower triangular structure with irregularly spaced pivots, but introduced a substantially higher level of within-column sparsity. As in Scenario 1, the first pivot was fixed at 1 and the remaining pivots were randomly generated between 2 and $m/5$, while the below-pivot entries were independently set to one with probability $0.5$. Both pivot locations and sparsity patterns varied across replicates. Here, $(m,H) \in \{(20,5), (50,10), (100,15)\}$.
Scenario 3 was designed to induce a block dependence structure.  Variables were partitioned into two blocks, and factors were evenly allocated across blocks, while imposing local positive lower triangular substructures within blocks. As in Scenario 2,  $(m,H) \in \{(20,5), (50,10), (100,15)\}$.
Scenario 4 considered a high-dimensional and highly sparse configuration. A generalized lower triangular structure was imposed with regularly spaced pivots. Below each pivot, the first $m/H$ entries of $\Delta$ were fixed to one, while all remaining entries were set to zero, thereby inducing highly localized factors with some overlap. In this case $(m,H) \in \{(20,8), (50,20), (100,40)\}$.
Whenever the below-pivot entries of $\Delta$ were generated randomly, the resulting configurations were adjusted to avoid empty rows and to satisfy the $3579$ counting rule \citep{fruehwirth2023GLT}. Conditional on $\Delta$, the nonzero entries of the loading matrix were generated from factor-specific Gaussian distributions; the idiosyncratic covariance matrix was set equal to the identity. Further details on the generation of $\Delta$ and $\Lambda$, including the treatment of practically vanishing loadings, are reported in the Supplementary Material.  
Figure~\ref{fig:examples} displays representative loading matrices $\Lambda_0$
generated under the four scenarios for $m=20$.

\begin{figure}[t!]
\centering
 \includegraphics[height=.44\textwidth]{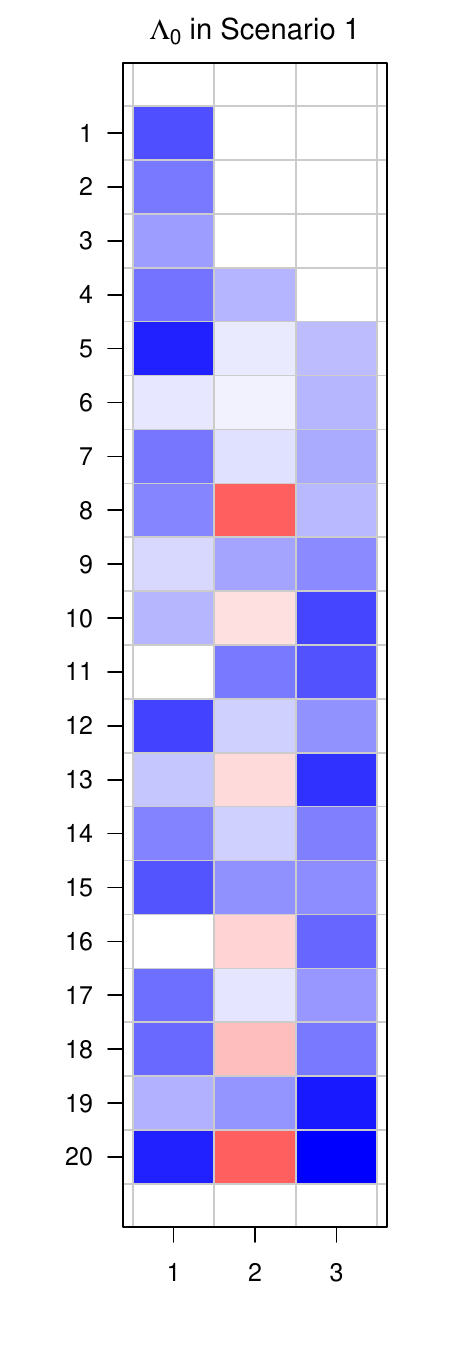}
 \includegraphics[height=.44\textwidth]{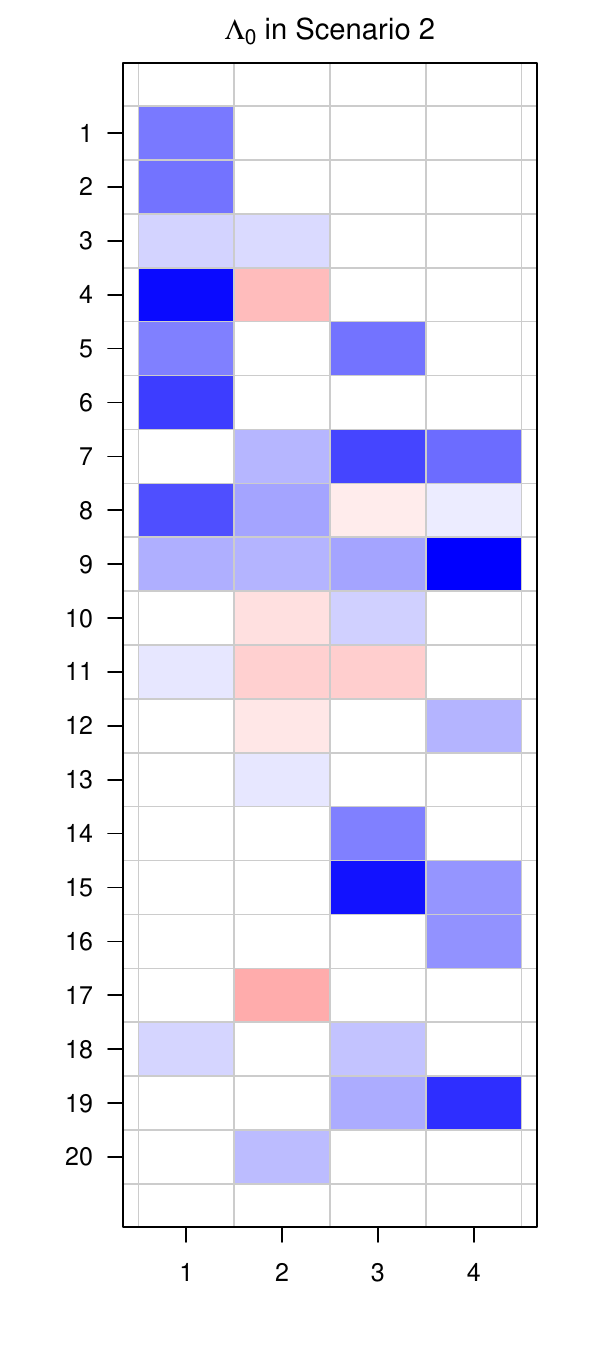}
 \includegraphics[height=.44\textwidth]{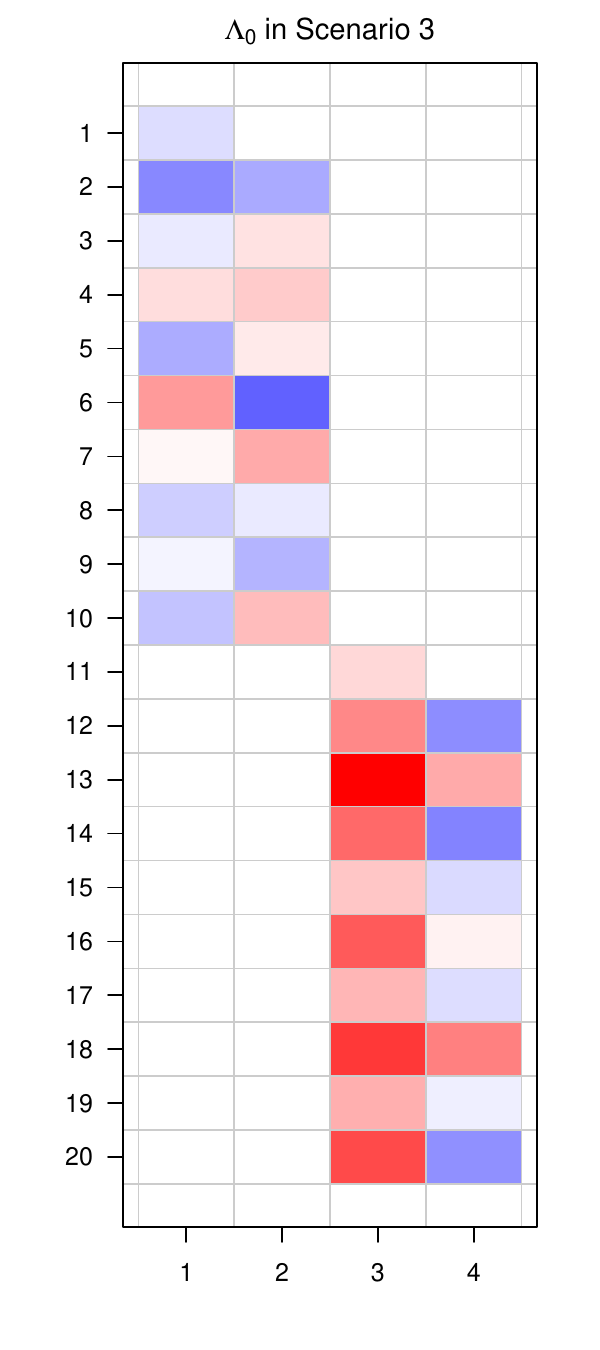}
 \includegraphics[height=.44\textwidth]{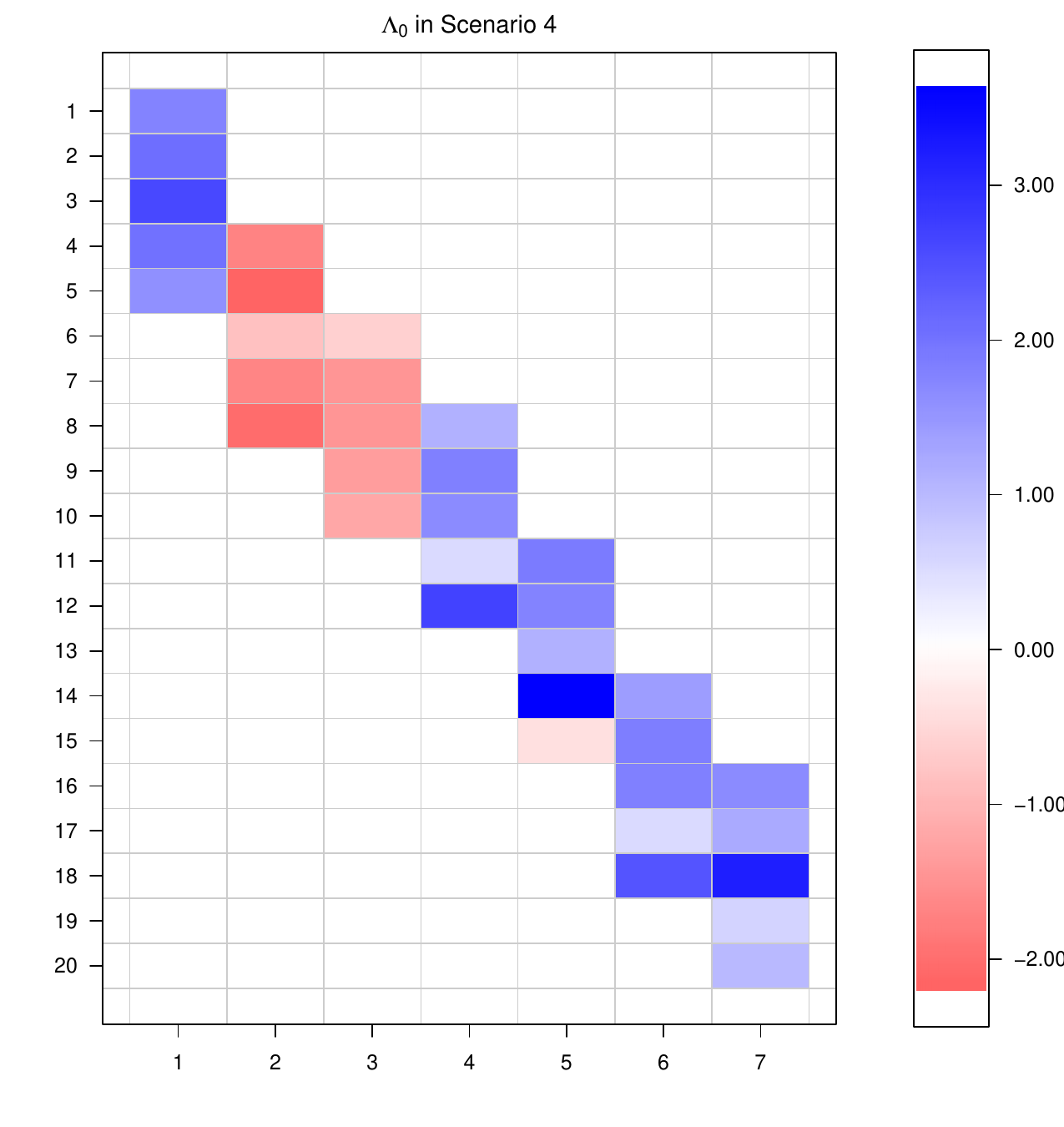}
\caption{Examples of true $\Lambda_0$ matrices under the four simulation scenarios with $m=20$}
\label{fig:examples}
\end{figure}

\begin{figure}[t!]
\centering
 \includegraphics[width=\textwidth]{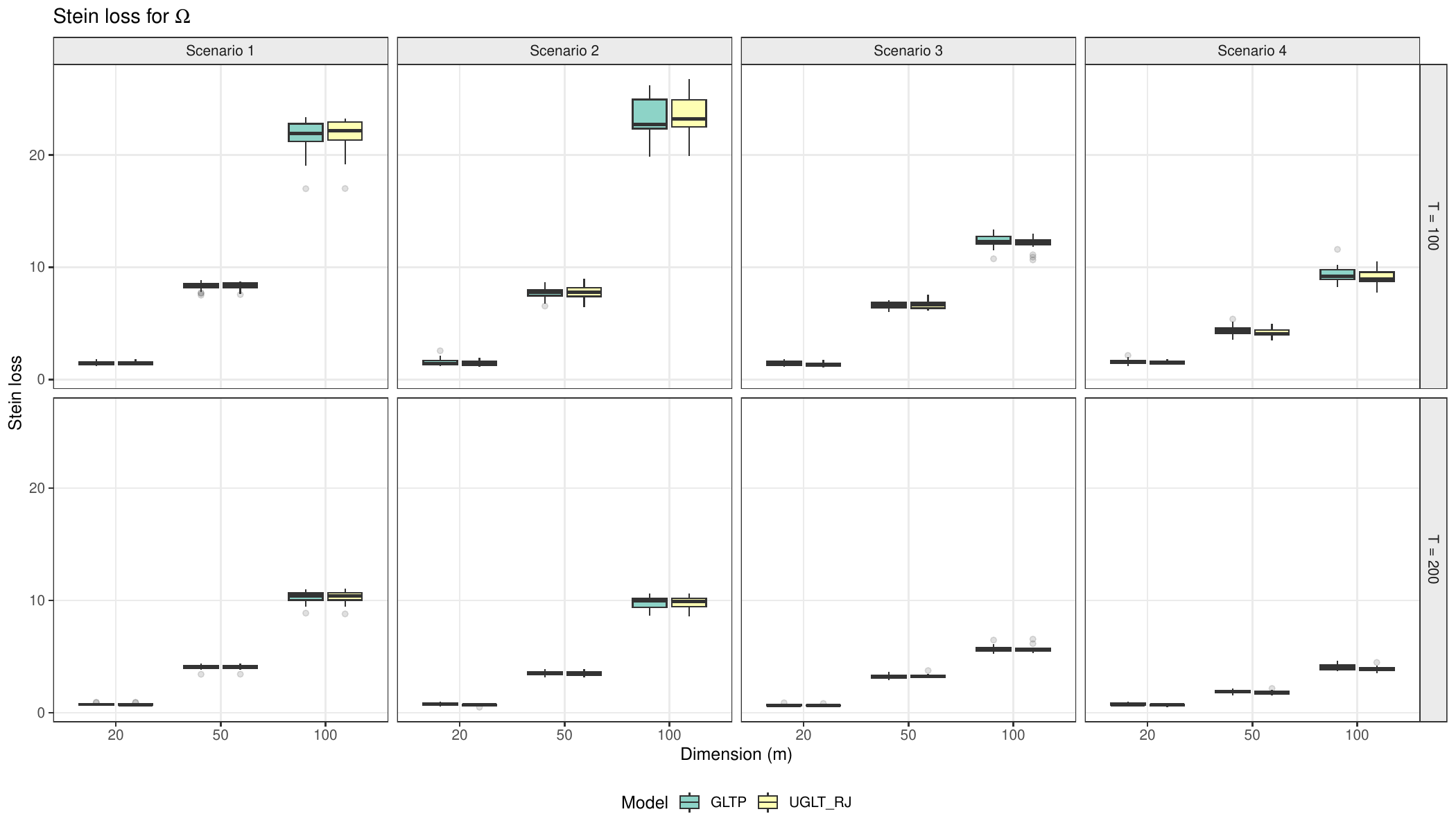}
 
 \includegraphics[width=\textwidth, page=4]{all_simulation_results.pdf}
 \caption{Stein loss and posterior probability assigned to the true number of factors, $\Pr(H=H_0\mid Y)$ across replicates for the proposed \glt{}  prior (GLTP) and the approach proposed in \citet{fruehwirth23BA} (UGLT\_RJ)}
\label{fig:simresults}
\end{figure}

Performance was assessed through four summaries targeting three distinct aspects: covariance estimation, recovery of the latent dimension, and sparsity recovery. Global covariance estimation was evaluated using the posterior expected entropy, or Stein, loss between the covariance matrix $\Omega$ and its true value
$\Omega_0$,
\[
L(\Omega,\Omega_0)
=
\operatorname{tr}(\Omega\Omega_0^{-1})
-
\log\det(\Omega\Omega_0^{-1})
-
m.
\]
For each replicate, the posterior expected loss was approximated by averaging
$L(\Omega,\Omega_0)$ over the admissible posterior draws, namely those satisfying the $3579$ counting rule \citep{fruehwirth2023GLT}.
Recovery of the latent dimension was evaluated through two complementary
summaries: the estimation error of the posterior median,
$\widehat H_{\mathrm{med}}-H_0$, and the posterior probability assigned to the true number $H_0$ of factors, $\Pr(H=H_0\mid Y)$. The former measures point-estimation
accuracy, whereas the latter also accounts for posterior uncertainty about the
latent dimension. Finally, recovery of the sparsity pattern of $\Delta$ was
assessed using the posterior inclusion probabilities
$\Pr(\delta_{jh}=1\mid Y)$ as classification scores for the corresponding true
entries of $\Delta_0$. Their ability to discriminate active loadings from
structural zeros was summarized by the area under the receiver operating
characteristic curve (AUC). Before computing the inclusion probabilities entering the AUC, posterior columns were aligned with the true loading structure; details of the alignment procedure are provided in the Supplementary Material.

Overall, the two approaches yielded remarkably similar results for covariance
and sparsity recovery. Across all scenarios, dimensions, and sample sizes, the Stein losses reported in the upper panel of Figure~\ref{fig:simresults} were essentially indistinguishable for GLTP and UGLT-RJ. As expected, covariance estimation improved with increasing sample size and became more challenging as $m$ increased. A similar picture emerged for sparsity recovery: the AUC values were uniformly high and nearly identical for the two procedures across all configurations; See Table~\ref{tab:auc} for details.

The lower panel of Figure~\ref{fig:simresults} shows a similarly close agreement in latent-dimension recovery for Scenarios~1--3. In these settings, the posterior median of $H$ almost always coincided with the true number $H_0$, and both procedures typically assigned posterior probability close to one to the true number of factors. 
A clearer difference emerged in Scenario~4, which combines a relatively large number of factors with a highly sparse and localized loading structure.
Here,  the approach of \citet{fruehwirth23BA} exhibited increasing posterior uncertainty about $H$ as $m$ increased, particularly for $T=100$, and occasionally underestimated the true latent dimension. In contrast, the proposed \glt retained substantially greater posterior concentration around $H_0$. The difference was attenuated for $T=200$, although it remained visible in the largest-dimensional configuration.

These results are particularly relevant in light of the different prior specifications required by the two approaches. The \glt{}  prior was
implemented throughout using simple uniform priors for the $\tau_j$s and $\nu$ ($\alpha=\beta=a_\nu=b_\nu=1$), whereas for the approach benchmark we adopted the hierarchical hyperprior specification for 
{$\alphaPB$ and $\gammaPB$ advocated by
\citet{fruehwirth23BA}  
with $\alphaPB \sim  \mbox{Ga}  (6, 3(\Hm-2)/\Hm) $ and $\gammaPB \sim 
\mbox{Ga}  (6, 6) $, where $\Hm = \lfloor (m-1)/{2}\rfloor$}. 
Thus, the \glt{}  prior essentially matches the performance of the unordered generalized lower triangular approach of \citet{fruehwirth23BA} in settings where the latter already performs extremely well, while relying on a considerably simpler prior specification, and provides greater robustness in the challenging sparse, high-rank configuration.
Together with the clearer prior elicitation discussed in Section~\ref{sec:GLTIBP}, these results highlight an important practical advantage of the proposed formulation.

\begin{table}
\caption{Median area under the curve for recovery of the sparsity pattern of $\Delta$ across simulation replicates for the \glt prior (GLTP) and the \citet{fruehwirth23BA} approach (UGLT-RJ)}
\centering
\begin{tabular}[t]{lccccc}
\toprule
& & \multicolumn{2}{c}{$T=100$} & \multicolumn{2}{c}{$T=200$} \\
\cmidrule(lr){3-4}\cmidrule(lr){5-6}
Scenario & $m$ & GLTP & UGLT-RJ & GLTP & UGLT-RJ\\
\midrule
Scenario 1 & 20 & 0.967 & 0.965 & 0.980 & 0.978\\
 & 50 & 0.920 & 0.920 & 0.947 & 0.949\\
 & 100 & 0.906 & 0.907 & 0.917 & 0.920\\
 \addlinespace
Scenario 2 & 20 & 0.974 & 0.970 & 0.985 & 0.985\\
 & 50 & 0.969 & 0.963 & 0.992 & 0.991\\
 & 100 & 0.974 & 0.969 & 0.996 & 0.996\\
 \addlinespace
Scenario 3 & 20 & 0.975 & 0.979 & 0.994 & 0.993\\
 & 50 & 0.985 & 0.987 & 0.992 & 0.994\\
 & 100 & 0.988 & 0.990 & 0.993 & 0.995\\
 \addlinespace
Scenario 4 & 20 & 0.990 & 0.986 & 0.999 & 0.999\\
 & 50 & 0.995 & 0.994 & 0.999 & 0.999\\
 & 100 & 0.991 & 0.989 & 0.997 & 0.997\\
\bottomrule
\end{tabular}
\label{tab:auc}
\end{table}

\section{Application to personality questionnaire data}
\label{sec:application}

We illustrate the proposed methodology using the Big Five Personality Test data,
obtained from the Open Psychometrics project \citep{Tunguz2019}.
The data contain responses to 50 personality items measured on a five-point Likert scale. The items were designed to measure five broad personality domains, namely extraversion, agreeableness, conscientiousness, neuroticism and openness,
with ten items associated with each domain \citep{goldberg1990}. 
This known organization provides a useful external benchmark for assessing the latent structure recovered by the model, without imposing the five-domain representation a priori.

The same data have been considered by \citet{ohn2022posterior} in the context of Bayesian estimation of factor dimensionality. Their analysis primarily used the canonical five-domain organization as a benchmark for comparing estimates of the number of factors and the associated loading matrices. Consistently with its illustrative purpose, however, that analysis relied on a number of simplifying choices, including treating the Likert responses as Gaussian and disregarding the available covariate information. Our aim is broader. We seek not only to learn the latent dimensionality, but also to characterize the sparse dependence structure among questionnaire items and the associated posterior uncertainty. To this end, we explicitly accommodate the discrete nature of the responses and exploit the available covariates to model systematic variation in the  mean, thereby separating response-process heterogeneity from the remaining latent dependence.

We restrict attention to respondents with complete answers to all 50 items and retain only records with a single response per IP address.
From the resulting sample, we randomly select 
$T=200$  individuals. The use of a moderate sample size is deliberate. Since the original database contains a very large number of respondents, an analysis based on the complete sample would largely suppress posterior uncertainty about low-dimensional features and would turn the example primarily into an exercise
in computational scalability. 
In contrast, $T=200$ provides a regime in which uncertainty about both factor dimensionality and sparse loading structure remains substantive. This choice also facilitates comparison with \citet{ohn2022posterior}, who considered subsamples of the same size. To investigate the sensitivity of the results to the particular subsample, we additionally repeat the analysis for 20 independently selected samples of the same size. Details are reported in the Supplementary Material.

We now carefully describe the modeling strategy for this dataset. For the  $t$-th subject ($t=1,\ldots,T$ ), let $y_t=(y_{1t},\ldots,y_{50,t})^{\mathrm T},$ 
denote the observed questionnaire responses. 
To deal with ordinal nature of the response, we embed the factor model within the simultaneous transformation and rounding framework of \citet{kowal2020simultaneous}. For a related approach in the context of single-cell RNA sequencing data analysis see \citet{canale2026structured}. 
Specifically, each observed response 
$y_{jt}$ is linked to a latent continuous variable  $z_{jt}$ through a rounding mechanism,
\begin{equation}
\label{eq:star1}
y_{jt}
=\ell
\quad\Longleftrightarrow\quad
\exp\{
z_{jt}
\} \in\mathcal A_\ell,
\qquad \ell=1,\ldots,5,
\end{equation}
where
$\mathcal A_1=(0,1),$
$\mathcal A_\ell=[\ell-1,\ell)$ for $\ell=2,3,4$, and 
$\mathcal A_5=[4,\infty).$
Thus, the rounding sets partition the positive real line at the cut points
$1,2,3,4$, while the logarithmic transformation maps the underlying continuous
response to the real line. See \citet{kowal2020simultaneous} for 
details and generalizations. 

In addition to the questionnaire responses, the data provide item-specific response times and information on the time spent on different stages of the questionnaire that we used as covariates in a linear regression for the mean of 
$z_{jt}$. 
Specifically, we use the mean response time within each of the five blocks of ten questions, together with the time spent on the introductory and final pages. Let 
$x_t$ 
denote this covariate vector after standardization and including an intercept. 
The latent variables 
$z_{jt}$ 
are then modeled via 
	\begin{equation}
    z_{jt} = x_t^T \beta_j + \varepsilon_{jt}, 
    \label{eq:star2}
\end{equation}
where we model the latent residual Gaussian noise 
$\varepsilon_t=(\varepsilon_{1t}, \ldots, \varepsilon_{50,t})^T$ as $\varepsilon_t = \Lambda \eta_t + \epsilon_t$ assuming for $\Lambda$, $\eta_t$, and $\epsilon_t$
the same specification as discussed in Sections~\ref{sec:intro}--\ref{sec:GLTstructure}. Additional vague Gaussian priors for 
the regression coefficients  $\beta_j$ 
complete the model specification.

\begin{figure}[t!]
\centering
 \includegraphics[height=.6\textwidth]{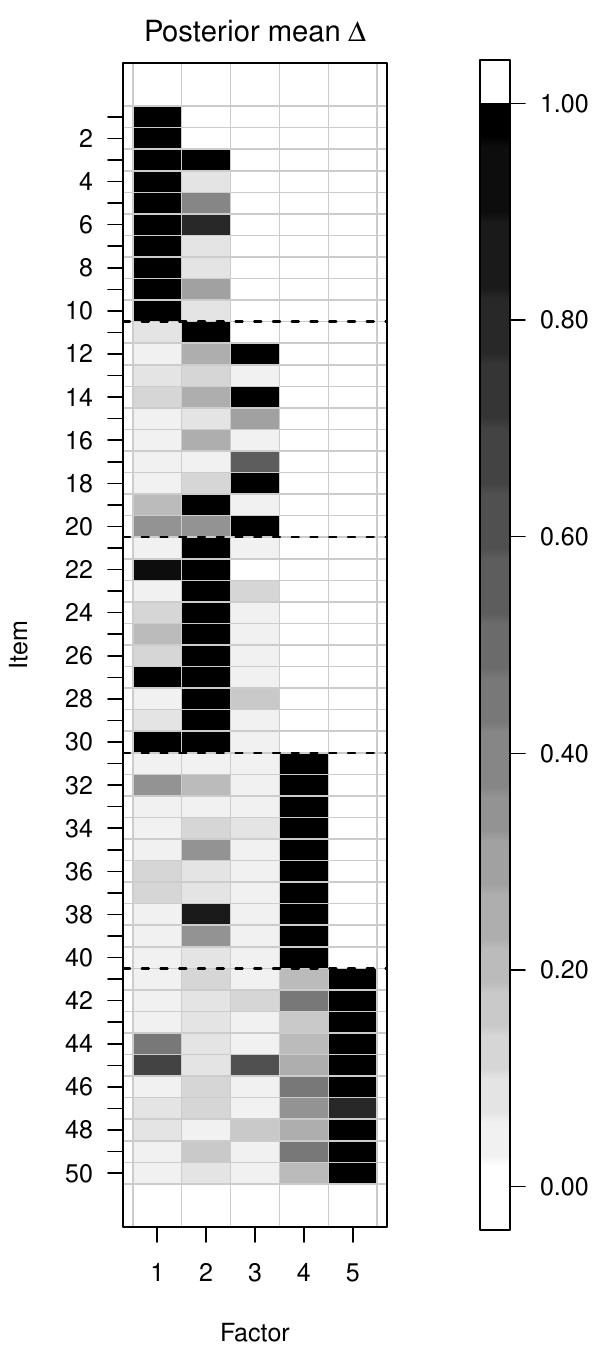}
 \includegraphics[height=.6\textwidth]{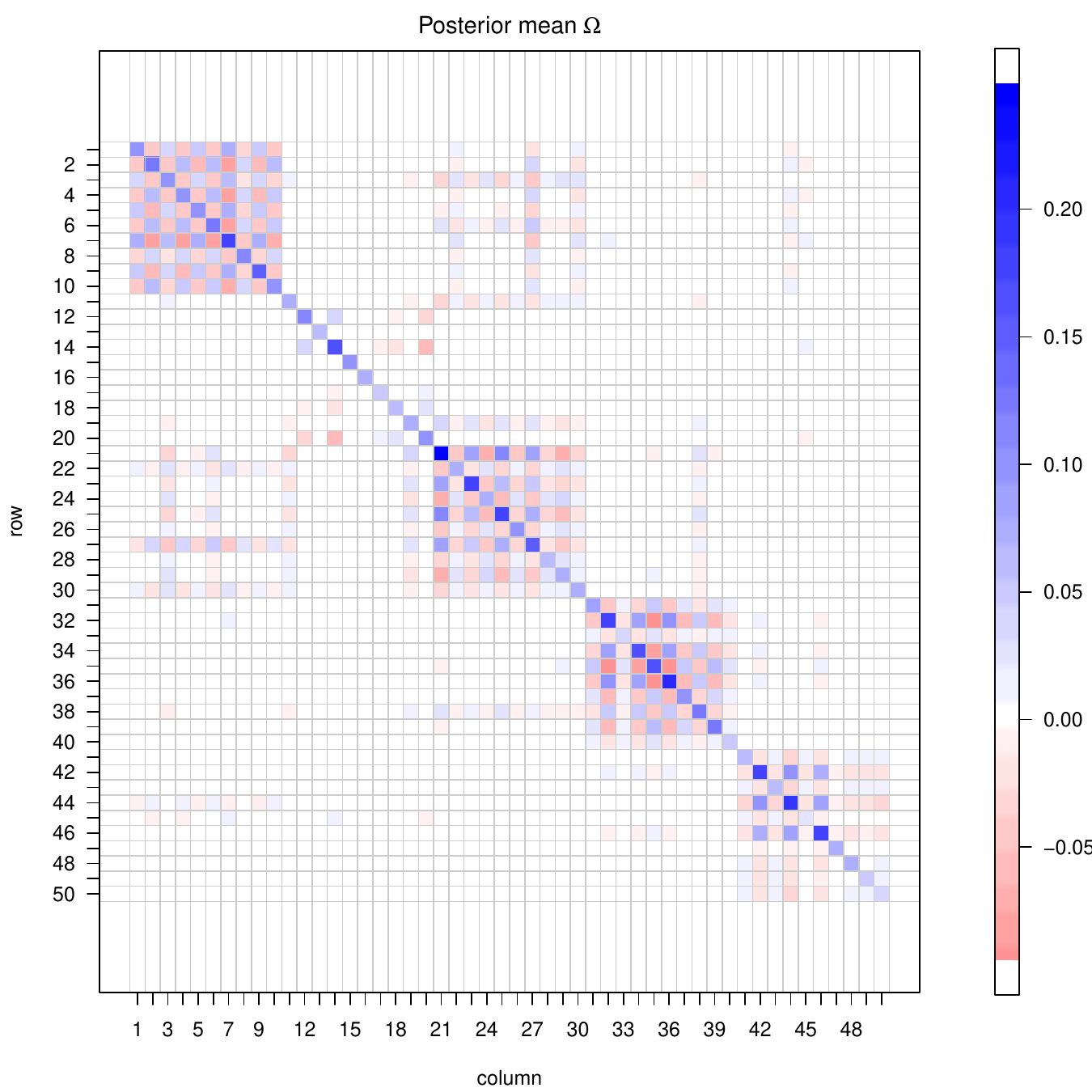}
\caption{Posterior mean $\Delta$ (left) and $\Omega$ (right) for the Big Five personality data. Dashed orizontal lines on the left plot divide the 50 items into the 5 blocks of questions.}
\label{fig:application}
\end{figure}

The posterior mean and mode of $H$ are $5.01$ and $5$, respectively, while the posterior probability of the canonical five-factor representation is $0.99$. Thus, the data provide strong support for a five-dimensional latent representation.

The \glt\ formulation additionally permits direct inference on the item--factor incidence structure. The left panel of Figure~\ref{fig:application} reports the posterior inclusion probabilities, obtained as entry-wise posterior means of the draws of $\Delta$ satisfying the identification criterion \citep{fruehwirth2023GLT} for the modal pivot configuration. The resulting pattern largely recovers the expected five-domain organization, while also revealing departures from an idealized simple structure. The first factor is almost exclusively associated with the first block of ten items, whereas the second factor is predominantly associated with the third block. The third factor is considerably sparser and is supported only by a small subset of items, mostly within the second block. The fourth and fifth factors are clearly associated with the fourth and fifth item groups, respectively. Hence, although the posterior strongly favours five latent dimensions, the evidence supporting the individual domains is heterogeneous: four factors display a pronounced block structure, whereas the second questionnaire domain is represented by a substantially sparser loading pattern.

The right panel of Figure~\ref{fig:application} reports the posterior mean of the residual covariance matrix $\Omega$. After accounting for the covariate-dependent latent mean and the factor component, the residual dependence is strongly concentrated within the first, third, fourth and fifth item blocks. In contrast, the second block exhibits comparatively weak residual association. This feature is particularly informative when compared with the empirical covariance matrix of the original Likert responses, reported in the Supplementary Materials, where the second block displays a noticeably richer pattern of positive and negative associations. The attenuation of these associations in the posterior residual covariance suggests that a substantial part of the raw dependence in this block is explained by the  representation induced by \eqref{eq:star1}-\eqref{eq:star2}, rather than persisting as residual dependence.

Further posterior characteristics including posterior mean and credible intervals for the factor sparsity and posterior mean inclusion probability for each factor and block of items are reported and discussed in the Supplementary Materials.

\section{Discussion}
\label{sec:discussion}

We have introduced a \glt\  prior for identifiable sparse Bayesian factor models with an unknown number of factors. The proposed construction separates the prior control of global dimensionality from that of within-factor sparsity, allowing these two features to be elicited through distinct and interpretable hyperparameters. At the same time, the \glt\ representation directly enforces an identifiable loading structure and avoids the need for reversible-jump schemes over competing pivot configurations. Posterior computation can be carried out through a relatively simple Markov chain Monte Carlo algorithm based on Gibbs updates and Metropolis--Hastings steps. The resulting framework therefore combines interpretability, flexibility and computational tractability within a single prior specification.

The simulation experiments show that the proposed prior is able to recover both the effective factor dimensionality and sparse loading patterns across a range of configurations. The application to the Big Five Personality Test further illustrates the inferential advantages of the approach and its ease of generalization.  The posterior strongly supports a five-dimensional representation, while revealing heterogeneous sparsity across the inferred factors and departures from an idealized block structure.

Several extensions appear promising. In the present formulation, the pivot indicators
are conditionally independent Bernoulli random variables with common probability $\nu$. When variable-specific auxiliary information 
$x_i$ 
is available, a natural generalization is to let the pivot probabilities depend on meta-covariates, for example letting $\operatorname{logit}(\nu_i)=\gamma_0+x_i^\top\gamma$.
This would allow prior information to influence not only the expected number of pivots, but also their locations. Such an extension is closely related to \citet{canaleschiavonstolf2026}, where auxiliary information is used to guide sparsity and pivot configurations. More generally, similar meta-regressions could be introduced for the post-pivot inclusion probabilities, yielding a fully structured \glt\  prior.


\section*{Acknowledgments}
Antonio Canale acknowledges support from the  Royal Society Wolfson Visiting Fellowship.


\appendix

\section{Appendix}

\renewcommand{\thesection}{\Alph{section}}
\renewcommand{\thetable}{\Alph{section}.\arabic{table}}
\renewcommand{\thefigure}{\Alph{section}.\arabic{figure}}
\renewcommand{\theequation}{\Alph{section}.\arabic{equation}}

\subsection{Computational details}

We derive some useful quantities.
The conditional joint probability of $y_t$ given $\eta_t$ and $\Lambda$, i.e. the likelihood given the latent factors and the factor loadings, 
is
\begin{align}
    \P(y_1, \dots y_T \mid \eta_1, \dots,\eta_T, \Lambda, \Sigma) \propto \exp\left\{
    -\frac{1}{2}\sum_{t=1}^T \sum_{i=1}^m   \frac{\left[ y_{it} - \lambda_{i\cdot } \eta_t \right]^2}{\sigma_i^2}
    \right\}.
\end{align}
{where $\lambda_{i\cdot }$ is the $i$th row of $\Lambda$.}
 Isolating the $i$-th entry we get 
\begin{align}
    \P({y_{i1}, \dots y_{iT}} 
    \mid \eta_1, \dots,\eta_T, \Lambda, \Sigma, \Delta) \propto 
    \exp\left\{
    -\frac{1}{2\sigma^2_i}  \left[ y_{i\cdot } -     
    {\eta_i^\delta} {\lambda_{i\cdot}^\delta}
    \right]^\top \left[ y_{i\cdot } - 
    {\eta_i^\delta} {\lambda_{i\cdot}^\delta}
    \right] 
    \right\}
\end{align}
which is the same of a regression model with
{outcome vector $y_{i\cdot}= (y_{i1}, \dots, y_{iT})^\top$ and  regression coefficients equal to the 
$q_i$-dimensional vector $\lambda_{i\cdot}^\delta$ containing the non-zero factor loadings  in row 
$\lambda_{i\cdot}$. The 
covariate matrix $\eta_i^\delta$ is derived from the corresponding latent factors with the $q_i$ columns being equal to  all factors $(\eta_{j1}, \ldots, \eta_{jT})^\top$,
 $j=1, \ldots, H$, for which $\delta_{ij}=1$.} 

Consistent with this, the marginal likelihood obtained by marginalizing out 
the {non-zero elements $\lambda_{i\cdot}^\delta$ of} $\lambda_{i\cdot}$ and $\sigma_i^2$ under prior \eqref{eq:priorlambdasigma} is
\begin{align} \label{app:marlik}
    \P({y_{i1}, \dots , y_{iT}}
    \mid \eta_1, \dots,\eta_T) = 
    \frac{\Gamma(a_\sigma +T/2)}{(2\pi)^{T/2}\Gamma(a_\sigma)} 
    \frac{b_\sigma^{a_\sigma}}{(b_\sigma + Q_i)^{a_\sigma +T/2}}
   \frac{\kappa^{-q_i/2}
}{
    | 
     {\kappa^{-1} I_{q_i} + (\eta_i^\delta) ^\top \eta_i^\delta }
|^{1/2}},
\end{align}
where $Q_i = 1/2 \left\{  
{y_{i\cdot} ^\top y_{i\cdot}
- y_{i\cdot} ^\top
\eta_i^\delta 
[\kappa^{-1} I_{q_i} + (\eta_i^\delta) ^\top \eta_i^\delta ]^{-1}
(\eta_i^\delta) ^\top
y_{i \cdot}}
\right\}.$

\subsection{Lemmas and proofs} \label{Proofs}

\begin{proof}[Proof of Lemma \ref{lem:priorsparsity}]
Define
\[
N_j = \sum_{i=\ell_j+1}^m \delta_{ij}, 
\quad {j=1, \ldots, H,}
\qquad N_m=H + \sum_{j=1}^H N_j.
\]
Moreover, let $K_i^*$ be a binary variable which indicates
whether row $i$ is a pivot row, 
so that $\sum_{i=1}^m K_i^*=H$.
We compute expectations in three nested steps.
 We first integrate out the $\delta_{ij}$ given the pivot locations $\{\ell_j\}_{j=1}^H$,  $\tau=(\tau_1,\ldots,\tau_H)$, and $H=h$. For each $j$ we have
\[
\mathbb{E} \, 
\! \big[N_j \mid \ell_j,\tau_j,H=h\big]
= \tau_j(m-\ell_j).
\]
Therefore,
\[
\mathbb{E}\left[N_m\,\middle|\,\{\ell_j\},\{\tau_j\},H=h\right]
= h +  \sum_{j=1}^h \tau_j(m-\ell_j).
\]
Then, we integrate out $\tau$ given the pivot sequence and $H$. Since $\mathbb{E}[\tau_j \mid {\alpha}]=\alpha/(\alpha+1)$ we have
\begin{equation*} 
\mathbb{E}\left[N_m\,\middle|  \,\{\ell_j\}, {\alpha}, H=h\right]
= h + \frac{\alpha}{\alpha+1} \sum_{j=1}^h (m-\ell_j).    
\end{equation*}
Finally, we have to integrate out the pivot locations. Before doing that, we note the identity
\begin{equation} \label{ljKi}
\sum_{j=1}^H (m-\ell_j) = \sum_{i=1}^m (m-i)K_i^*.
\end{equation}
Notably, the left hand side of equation above is a random quantity depending both on the random $H$ and random pivot locations. The right hand side, instead, is a function of the sequence $K_i^*$ only. We henceforth focus on the latter. 
By exchangeability of pivot locations given $H=h$, each row has the same conditional probability to be a pivot, hence $\mathbb{E}[K_i^*\mid H=h]=h/m$ for every $i$, and thus
\[
\mathbb{E}\!\left[\frac{K_i^*}{H}\,\middle|\, H=h\right] = 
\frac{1}{m}.
\]
If follows,
\begin{equation} \label{resultA}
    \mathbb{E}\left[\frac{1}{H}\sum_{i=1}^m (m-i)K_i^* \,\middle|\,H=h\right]
= \sum_{i=1}^m (m-i)\,\mathbb{E}\!\left[\frac{K_i^*}{H}\,\middle|\,H=h\right]
= \sum_{i=1}^m (m-i)\cdot\frac{1}{m} = \frac{m-1}{2}.
\end{equation}
Substituting back gives, for any $h\ge 1$,
\[
\mathbb{E}\!\left[\frac{N_m}{mH}\,\middle|\, {\alpha}, H=h\right]
= \frac{1}{m} + \frac{\alpha}{\alpha+1}\,\frac{m-1}{2m}.
\]
\end{proof}

\begin{proof}[Proof of Lemma \ref{lem:column_sparsity_var}]
Use the standard Beta--Binomial variance formula with $a=\alpha\beta$, $b=\beta$ and
$n=A_j$,
\[
\mathrm{Var}(N_j\mid A_j)
=
A_j\frac{\alpha(\alpha\beta+\beta+A_j)}{(\alpha+1)^2(\alpha\beta+\beta+1)}.
\]
Divide by $A_j^2$ to get $\mathrm{Var}(R_j)=\mathrm{Var}(N_j)/A_j^2$ and obtain 

\[
\mathrm{Var}(R_j\mid A_j)
=
\frac{\alpha}{A_j(\alpha+1)^2}\,\frac{(\alpha+1)\beta + A_j}{(\alpha+1)\beta+1}.
\]
Noting $(\alpha+1)\beta + A_j=( (\alpha+1)\beta+1) + (A_j-1)$ and expanding yields the
claimed decomposition
\[
\mathrm{Var}(R_j\mid A_j)
=
\frac{\alpha}{A_j(\alpha+1)^2}
+
\frac{\alpha(A_j-1)}{A_j(\alpha+1)^2\{(\alpha+1)\beta+1\}}.
\]
\end{proof}

\begin{proof}[Proof of Theorem~\ref{thm:uniformpivots}]
Let $K_i^*$ be a binary variable which indicates
whether row $i$ is a pivot row. By construction, conditionally on \(\nu\), the \(K_i^*\)'s are i.i.d. Bernoulli\((\nu)\). Hence any configuration with exactly \(H\) pivots has probability $\nu^H(1-\nu)^{m-H}.$
Therefore, conditionally on \(H\), all \(\binom{m}{H}\) possible pivot configurations are equally likely, so that
\[
	\P(\ell_1=l_1, \ldots, \ell_{H}=l_H \mid  H )
=
\binom{m}{H}^{-1}
=
\frac{H!(m-H)!}{m!}.
\]
\end{proof}

\subsection{Additional results on the \citet{fruehwirth23BA} prior}
\label{UGLTprior}

The prior proposed by \citet{fruehwirth2023GLT}
operates with an indicator  matrix 
$\Delta$ of fixed dimension $m \times \Hm$, where the number of columns $\Hm\leq \lfloor (m-1)/2\rfloor$ is a finite upper bound for the factor dimension $H$.
 It is  inspired by a  finite version of two-parameter beta prior by \citet{gha-etal:bay}, where in each column $j$ a priori 
$ \P(\delta_{ij} = 1 |\tau_{j} ) = \tau_{j} $ and  $\tau_j$ are independent and identically distributed beta random variable, $\tau_j \sim \mbox{Beta}(\alphaPB/\Hm \gammaPB, \gammaPB)$ with hyperparameters $\alphaPB$ and $\gammaPB$.   
 The latter modifies the prior of \citet{gha-etal:bay} so that all realizations from $\P(\Delta \mid \alphaPB, \gammaPB, \Hm)$ have an unordered generalized lower triangular structure. The pivots $\ell_1, \ldots, \ell_H$ of the $H$ non-zero columns occupy distinct rows without any ordering constraint and have a uniform prior distribution in the sense of Theorem~\ref{thm:uniformpivots}.
 Thus, unlike in our approach, the same hyperparameters $\alphaPB$ and $\gammaPB$ in the prior $\tau_j \sim \mbox{Beta}(\alphaPB \gammaPB/\Hm, \gammaPB)$ govern both row and column sparsity in $\Delta$.

 Deriving a feature allocation representation
 of the prior used in \citet{fruehwirth23BA}  as in Section~\ref{sec:GLTIBP}
is a helpful in understanding the impact of $\alphaPB$ and $\gammaPB$ on row sparsity. 
The sparsity pattern below the pivots is controlled
by the conditional prior $ \P(\delta_{ij} = 1 \mid \ell_j, \tau_j) = \tau_{j} $ for $i>\ell_j$.
As before, let $K_i$ denote the number of non-zero columns in the first $i$ rows of $\Delta$, and let $d_{ij} = \sum_{l=\ell_j+1}^i \delta_{lj}$ count the customers who have chosen dish $j$ after its first selection. Under the common culinary metaphor, customer $i+1$ chooses dish $j$ among the $K_i$ previously selected dishes with probability $\tau_j$, whose conditional distribution is
\[
\tau_j \mid \ell_j, d_{ij}, \alphaPB, \gammaPB, \Hm  \, \sim\,   \mbox{Beta}\left(\alphaPB/\Hm \gammaPB + d_{ij}, \gammaPB + i-\ell_j -d_{ij}\right).
\]
 Its expectation is  
 \[\frac{\alphaPB/\Hm \gammaPB + d_{ij}}{ (\alphaPB/\Hm +1) \gammaPB + i-\ell_j}.
 \]
 We find that for a given $\gammaPB$, increasing $\alphaPB$
adds to the popularity of dish $j$, while increasing $\gammaPB$ for a given $ \alphaPB$ has a strong  averse effect. 
As  $\Hm$ increases, the expected probability
for dish $j$ converges to $d_{ij}/  (\gammaPB + i-\ell_j)^{-1}$, hence $\gammaPB$ can be regarded
as a prior  penalty in terms of imaginary times  dish $j$ was not chosen by  $\gammaPB$  costumers.

In addition,  customer $i+1$  has a positive probability to choose a  single new dish among the remaining $\Hm-K_i$ dishes as long as $K_i< \Hm$. 
 This probability has no simple expression  due to the subtle impact of the generalized lower triangular constraint imposed on the  rows
of $\Delta$ for any given $K_i< \Hm$.  
The same constraint allows each customer to choose at most one new dish, in sharp contrast to the unconstrained two-parameter beta prior, under which customer $i+1$ may choose any number of the remaining dishes.

 
 To understand how the hyperparameter 
 $\alphaPB$ and $\gammaPB$ impacts column sparsity, 
 it is helpful to consider a  representation where we move  from one dish to next, increasing $j=1, \ldots, \Hm$ sequentially. 
 Under the two-parameter beta prior, there is a positive probability that column $j$ is a zero column, meaning that  none of the $m$ customers chooses dish $j$, or, equivalently,   $d_j=\sum_{i=1}^m \delta_{ij}=0$.
 Let $\Lq{j}$ denote the number of dishes that where chosen by at least one costumer among the $j$ dishes offered sofar, with $\Lq{0}=0$ by definition.
Then,  given $\Lq{j-1}= \jr-1$,
 \begin{equation} \label{djzero}
\pq{\jr}=\P(d_{j}=0| \Lq{j-1}=\jr-1, \alphaPB, \gammaPB,
 \Hm ) = \frac{B(\alphaPB/\Hm \gammaPB, \gammaPB + m -\jr+1)}{B(\alphaPB/\Hm \gammaPB, \gammaPB)}. 
 \end{equation}
A proof is given below. 

Based on the sequence $\pq{\jr}, \jr=1, \ldots,h+1$, the prior probability  $\P(H=h|\alphaPB, \gammaPB,
 \Hm)$ is given by:
 \begin{eqnarray}
     \P(H=h|\alphaPB, \gammaPB,
 \Hm)=   \prod_{\jr=1}^h \pq{\jr}^{\tilde{W}_\jr} (1-\pq{\jr}) \, (\pq{h+1})^{m-h- \sum_{\jr=1}^h\tilde{W}_\jr} ,
 \end{eqnarray}
where $\tilde{W}_\jr$ is the waiting time between jump $\jr-1$ and jump $\jr$ of the process $ \Lq{j} $. 
No closed-form solution is available for the expectation 
$\mathbb{E} (H|\alphaPB, \gammaPB,
 \Hm)$, but column sparsity obviously depends on both hyperparameters  
$\alphaPB$ and $\gammaPB$. 


Under the assumption  $\Hm = \lfloor \frac{m-1}{2}\rfloor$  suggested by  \citet{fruehwirth2023GLT},  
it possible to derive the 
following approximation of the prior expectation of $H$
for large $m$:
 \begin{eqnarray} \label{qapprox}
     \mathbb{E} (H|\alphaPB, \gammaPB,
 \Hm) \approx  \alphaPB \gammaPB \log m - \alphaPB \gammaPB \psi(\gammaPB),
 \end{eqnarray}
 where $\psi(x)=\Gamma'(x)/\Gamma(x)$ is the digamma function. A proof is given below. 
 
Interestingly, the first term of \eqref{qapprox} coincides with the limiting expectation $\mathbb{E}(H \mid \alphaPB, \gammaPB)$ under the two-parameter beta prior as $\Hm \to \infty$. The second term accounts for the reduction in the expected  factor dimension  induced by the generalized lower triangular constraint, which allows to each customer to choose at most one new dish.

Finally, the associated prior expected non-sparsity equals 
\begin{eqnarray} \label{lemma4}
 \mathbb{E}\left[\frac{N_m}{mH} \middle| \alphaPB, \Hm\right] = \frac{1}{m} + \frac{\alphaPB}{\alphaPB + \Hm} \frac{m-1}{2m}.   
\end{eqnarray}
Its proof is reported below

\begin{proof}[Proof of Eq.~\eqref{djzero}]
For the  first dish ($j=1$), the probability of a zero column is equal to
 \begin{equation*}
\pq{1}=\P(d_1=0|  \alphaPB, \gammaPB,
 \Hm ) = \frac{B(\alphaPB/\Hm \gammaPB, \gammaPB + m )}{B(\alphaPB/\Hm \gammaPB, \gammaPB)}.
\end{equation*}
With probability $1-\pq{1}$,
the first dish will be chosen  at least once
and $\Lq{1}= \Lq{0}+1$ increases by 1. 
As long as $\Lq{j}=0$ remains unchanged, 
 the probability that any subsequent dish $j +1$ is never chosen is equal to  $\pq{1}$. With 
probability $1-\pq{1}$,
dish $j$ will be chosen  at least once
and $\Lq{j}= \Lq{j-1}+1$ increases by 1. 
In case $\Lq{j}= \Lq{j-1}+1$ increases from 0 to 1,  the probability $\pq{2}$ that the subsequent  dish $j+1$ will not be chosen changes due to the generalized lower triangular constraint imposed on column $j$. For any  possible pivot $\ell_1$ in column $j$, by definition $\delta_{\ell_1,j+1}=0$ is a structural zero and only the $m-1$ indicators   $\delta_{i,j+1}$ where $i \neq l_1$ are subject 
to variable selection. 
The probability  $\pq{2}$ remains  constant until $\Lq{j}$ increases again. This process is repeated until $j=\Hm$ reaches the upper bound  and  $\Lq{m}$ defines the number of non-zero columns  as $H=\Lq{m}$.
Given $\Lq{j-1}= \jr-1$,
 \begin{equation*}
\pq{\jr}=\P(d_{j}=0| \Lq{j-1}=\jr-1, \alphaPB, \gammaPB,
 \Hm ) = \frac{B(\alphaPB/\Hm \gammaPB, \gammaPB + m -\jr+1)}{B(\alphaPB/\Hm \gammaPB, \gammaPB)},
\end{equation*}
since  all possible pivots 
$\ell_1 < \ldots < \ell_{\jr-1}$ in the current 
$\jr-1$  non-zero columns by definition impose the structural zeros 
$\delta_{ij}=0$ for all rows $i \in \{\ell_1 ,\ldots, \ell_{\jr-1}\}$ and only the remaining $m-\jr+1$ indicators  in column $j$ are subject 
to variable selection.
\end{proof}

\begin{proof}[Proof of Eq.~\eqref{qapprox}]
Note that $1 -\pq{\jr}$ can be expressed 
with the help of  two  ratios of gamma functions:
\begin{eqnarray*}
1 -\pq{\jr}= 1- 
\frac{\Gamma (\gammaPB + m -\jr+1)}{\Gamma(\alphaPB/\Hm \gammaPB + \gammaPB + m -\jr+1)}
\frac{\Gamma (\alphaPB/\Hm \gammaPB+ \gammaPB)}{\Gamma( \gammaPB)}.
\end{eqnarray*}
Using that $\Gamma(x)/\Gamma(x+\Delta x) \approx x^{-\Delta x}$, we obtain for large $m$ with $x=\gammaPB + m -\jr+1 \approx m$  and $\Delta x=\alphaPB \gammaPB/\Hm  $:
\begin{eqnarray*}
\frac{\Gamma (\gammaPB + m -\jr+1)}{\Gamma( \gammaPB + m -\jr+1 + \alphaPB/\Hm \gammaPB )} \approx  m^{-\alphaPB \gammaPB/\Hm}=
\exp(-\alphaPB \gammaPB/\Hm \log m) \approx 
1- \alphaPB \gammaPB/\Hm \log m.
\end{eqnarray*}
The second ratio can be approximated via a Taylor expansion around $\gammaPB$:
\begin{eqnarray*}
\frac{\Gamma ( \gammaPB + \alphaPB/\Hm \gammaPB)}{\Gamma( \gammaPB)} \approx 
\frac{\Gamma ( \gammaPB) + \alphaPB/\Hm \gammaPB \Gamma'(\gammaPB)}{\Gamma( \gammaPB)} =
1+ \alphaPB/\Hm \gammaPB \psi( \gammaPB).
\end{eqnarray*}
Ignoring terms that are quadratic in  $1/\Hm$,  we obtain
the asymptotic approximation  
\begin{eqnarray*} %
\qtilde \approx \frac{\alphaPB \gammaPB}{\Hm} (\log m - \psi(\gammaPB))
\end{eqnarray*}
 for 
the success probability $1 -\pq{\jr}$ for large $m$.
Based on $\qtilde$, the prior $H |\alphaPB, \gammaPB,
 \Hm$ is approximated for large $m$ by the binomial distribution  $H |\alphaPB, \gammaPB,
 \Hm \sim \mbox{Bin}(\Hm, \qtilde)$, with prior expectation $\Hm, \qtilde$. 
Multiplying $\qtilde$ with $\Hm$ yields  Eq.~\eqref{qapprox}.
\end{proof}

\begin{proof}[Proof of \eqref{lemma4}]
Since $\mathbb{E}[\tau_j \mid \alphaPB, \Hm]=\alphaPB/(\alphaPB+\Hm)$,
we have 
\[
\mathbb{E}\left[N_m\,\middle|  \,\{\ell_j\}, \alphaPB, \Hm, H=h\right]
= h + \frac{\alphaPB}{\alphaPB+\Hm} \sum_{j=1}^h (m-\ell_j),
\]
with the same arguments in the proof of Lemma \ref{lem:priorsparsity}.
Under the \citet{fruehwirth2023GCUSP} prior,  each row has the same conditional probability to be a pivot given $H=h$, hence  
\eqref{resultA} holds as well. Exploiting in addition 
the identify \eqref{ljKi}, we obtain \eqref{lemma4}.
\end{proof}

\bibliographystyle{apalike}
\bibliography{biblio.bib}

\newpage 
\section{Supplementary Materials}

\subsection{Additional details on Section 4}

\subsubsection{Synthetic data simulation}

We provide here additional details on the data-generating mechanism used in the simulation study. For each replicate, data were generated from the Gaussian factor model in (2), with loading matrix $\Lambda\in\mathbb{R}^{m\times H}$ and idiosyncratic covariance matrix equal to the identity. 

In Scenario~1, $\Delta$ followed a generalized lower triangular structure with irregularly spaced pivots. The first pivot was fixed at row~1, while the remaining pivot positions were generated randomly within the upper fifth of the rows. Conditional on the pivot positions, each admissible below-pivot entry was independently set equal to one with probability $0.95$. Thus, although the GLT restriction determined which entries could be nonzero, the precise sparsity pattern varied from replicate to replicate. This scenario was intended to represent a setting with relatively few and predominantly dense factors. The values of $(m,H)$ were $
(20,3), (50,8),$ and $(100,10).$

Scenario~2 used the same mechanism for generating irregular pivot positions, but the admissible entries below each pivot were independently activated with probability $0.5$. This generated substantially sparser factors, while retaining considerable heterogeneity in both pivot positions and within-column sparsity across replicates. The corresponding configurations were $(20,5), (50,10),$ and $(100,15).$

Scenario~3 was constructed to generate a block-structured loading matrix. The variables were divided into two blocks and the $H$ factors were allocated between them as evenly as possible. Within each block, the corresponding submatrix of $\Delta$ was assigned a positive lower triangular structure, while entries linking a factor to variables outside its associated block were fixed at zero. Hence, the resulting covariance structure was characterized by two groups of variables driven by separate sets of latent factors. The values of $(m,H)$ were
 $(20,5), (50,10),$ and $(100,15),$ as in Scenario 2. 
 
Finally, Scenario~4 was designed to represent a more challenging high-dimensional and highly sparse configuration. Pivots were placed regularly across the rows of $\Delta$. Starting from each pivot, only a short sequence of approximately $m/H$ entries was set equal to one, whereas all remaining admissible entries in that column were fixed at zero. The resulting columns were therefore highly localized, with only limited overlap among neighbouring factors. At the same time, the number of latent factors was considerably larger than in the other scenarios:  $(20,8), (50,20),$ and $(100,40).$

Conditional on $\Delta$, the nonzero entries of $\Lambda$ were generated using factor-specific location parameters. Specifically, independently for $h=1,\ldots,H$,

$$
\mu_h\sim
0.2\,\delta_{-1.5}
+
0.4\,\delta_{0}
+
0.4\,\delta_{1.5},
$$
and, for every active entry,
$$
\lambda_{jh}\mid \delta_{jh}=1,\mu_h
\sim N(\mu_h,0.5).
$$
Equivalently,
$$
\lambda_{jh}
=
\delta_{jh}\bigl(\mu_h+\varepsilon_{jh}\bigr),
\qquad
\varepsilon_{jh}\overset{\mathrm{iid}}{\sim}N(0,0.5).
$$
This construction produces columns whose nonzero entries may be centered around zero or display predominantly positive or negative values. When $\mu_h\geq0$, the corresponding pivot loading was constrained to be positive by taking its absolute value. No analogous sign constraint was imposed when $\mu_h=-1.5$.

To avoid active entries having a negligible magnitude relative to the overall scale of their row, a final adjustment was applied to the generated loading matrix. For row $j$, define

$$
v_j=1+\sum_{\ell=1}^H\lambda_{j\ell}^2.
$$
A nonzero loading was classified as practically vanishing whenever
$
|\lambda_{jh}|/v_j<0.02.
$
Whenever this condition was met, its magnitude was increased while preserving its sign, according to

$$
\lambda_{jh}
\leftarrow
\left\{\min_{r=1,\ldots,m}v_r+10^{-6}\right\}
\operatorname{sign}(\lambda_{jh}).
$$

Thus, active entries identified as too small were deliberately replaced by loadings of non-negligible magnitude, preventing the simulation design from containing nominally active entries carrying essentially no signal.

The idiosyncratic covariance matrix was fixed to
the identity  in all scenarios.

\subsubsection{Metrics}

Before computing the area under the curve (AUC) metric, we align the generic $s$th MCMC draw of $\Delta_s$, with $H_s$ columns, to the true matrix $\Delta_0$, with $H_0$ columns, using a greedy matching procedure. We visit the first $\min(H_s,H_0)$ columns of $\Delta_0$ in their original order and match each to the unused column of $\Delta_s$ with the largest number of matching entries, counting both zeros and ones. Ties are resolved by selecting the first available column. The aligned matrix has $H_0$ columns: if $H_s<H_0$, its remaining columns are set to zero; if $H_s>H_0$, unmatched estimated columns are discarded. 
\newpage

\subsection{Additional details on Section 5}

\begin{table}[ht]
\centering
\caption{Posterior summaries for the 20 random subsamples of the original data. }
\begin{tabular}{rrrrrr}
  \hline
 Replicate \# & $\mathbb{E}(H \mid y)$ & Posterior mode & $\P(H=5\mid y)$ & $\mathbb{E}\left[\frac{N_m}{mH}\mid y\right]$ \\ 
  \hline
1 & 6.05 & 6.00 & 0.33 & 0.26 \\ 
  2 & 5.95 & 6.00 & 0.09 & 0.27 \\ 
  3 & 5.02 & 5.00 & 0.99 & 0.31 \\ 
  4 & 5.76 & 5.00 & 0.62 & 0.28 \\ 
  5 & 5.01 & 5.00 & 0.99 & 0.25 \\ 
  6 & 5.44 & 5.00 & 0.58 & 0.25 \\ 
  7 & 5.01 & 5.00 & 0.99 & 0.26 \\ 
  8 & 5.06 & 5.00 & 0.94 & 0.28 \\ 
  9 & 5.88 & 6.00 & 0.16 & 0.27 \\ 
  10 & 5.01 & 5.00 & 0.99 & 0.26 \\ 
  11 & 5.01 & 5.00 & 0.99 & 0.28 \\ 
  12 & 5.09 & 5.00 & 0.92 & 0.26 \\ 
  13 & 5.01 & 5.00 & 0.99 & 0.27 \\ 
  14 & 5.63 & 6.00 & 0.38 & 0.24 \\ 
  15 & 5.80 & 6.00 & 0.27 & 0.26 \\ 
  16 & 5.26 & 5.00 & 0.75 & 0.37 \\ 
  17 & 5.12 & 5.00 & 0.89 & 0.24 \\ 
  18 & 5.02 & 5.00 & 0.98 & 0.31 \\ 
  19 & 5.39 & 5.00 & 0.61 & 0.25 \\ 
  20 & 5.01 & 5.00 & 0.99 & 0.27 \\ 
   \hline
\end{tabular}
\end{table}

\begin{figure}[t]
\centering
\includegraphics[height=.6\textwidth, page=2]{empiricalVariance.pdf}
\includegraphics[height=.6\textwidth]{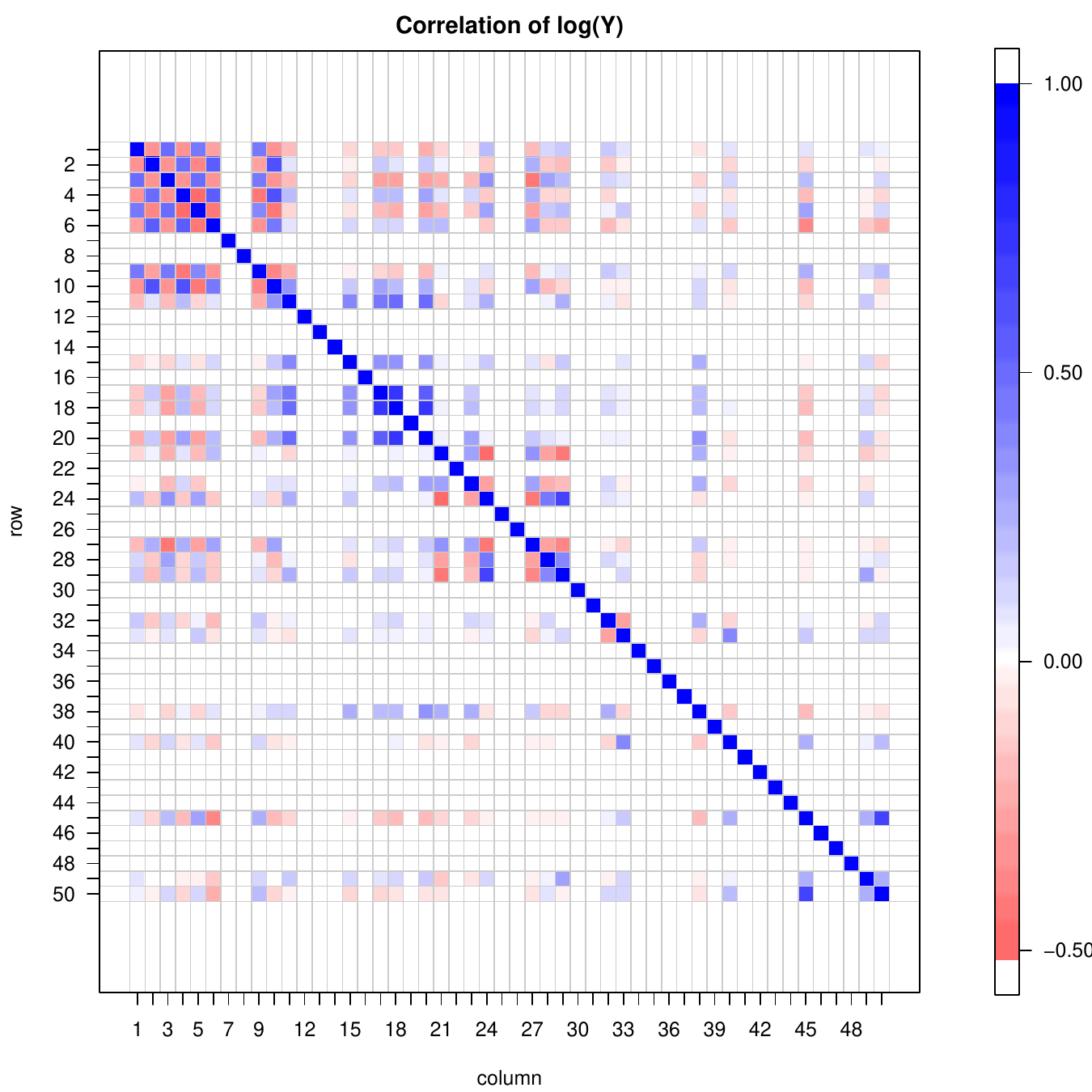} 
\caption{Empirical correlations computed on raw $Y$ (upper panel) and $\log(Y)$ (lower panel). }
\label{fig:empirical}
\end{figure}

\begin{figure}[t]
\centering
\includegraphics[width=.9\textwidth]{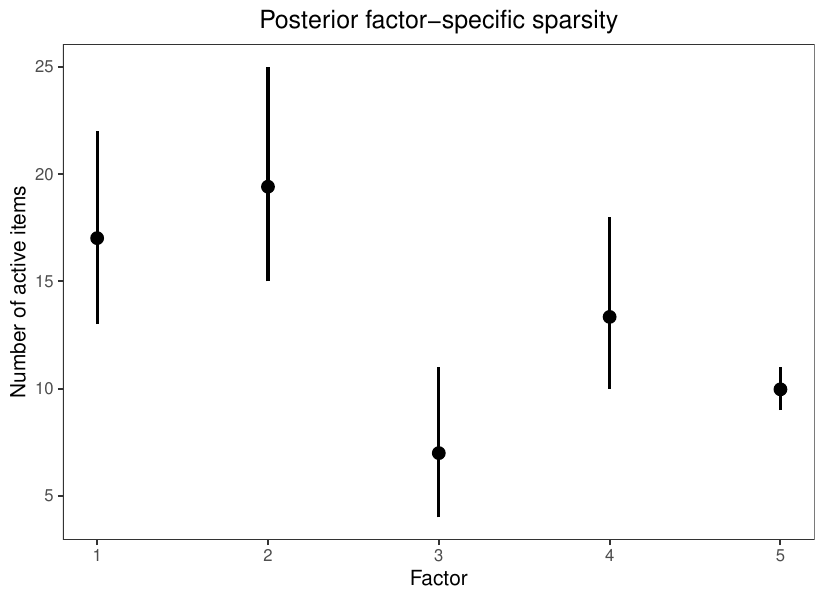}
\caption{Posterior mean (and 95\% credible intervals) for the factor-specific sparsity for the analysis reported in Section~5.}
\label{fig:facsparsity}
\end{figure}

\begin{figure}[t]
\centering
\includegraphics[width=.9\textwidth]{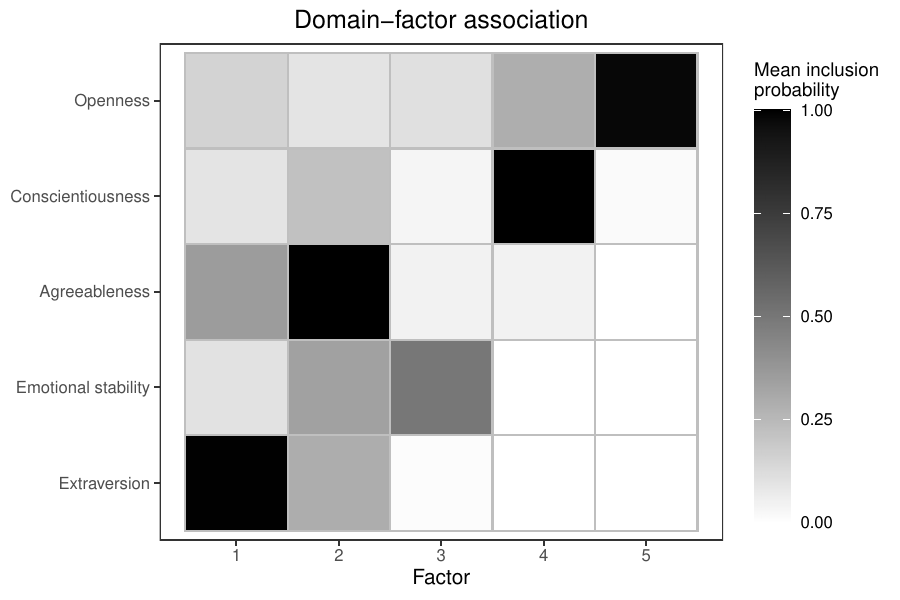}
\caption{Posterior mean inclusion probability for each factor and block of items for the analysis reported in Section~5.}
\label{fig:facitem}
\end{figure}



\end{document}